\pdfoutput=1

\documentclass[12pt]{article}

\usepackage{amsmath}
\usepackage{graphicx, psfrag, epsf}
\usepackage{enumerate}
\usepackage{natbib}
\usepackage{url} 
\usepackage{amsfonts}
\usepackage{graphicx}
\usepackage{amssymb}
\usepackage{amsthm}
\usepackage{verbatim}
\usepackage{epstopdf}
\usepackage{enumerate}
\usepackage{bm}
\usepackage{placeins}
\usepackage{subfigure}
\usepackage{wrapfig}
\usepackage{natbib}
\usepackage[english]{babel}
\usepackage{soul}

\newcommand\floor[1]{\left\lfloor #1 \right\rfloor}

\newtheorem{Theorem}{Theorem}[section]
\newtheorem{Lemma}[Theorem]{Lemma}

\newtheorem{Defn}[Theorem]{Definition}
\theoremstyle{remark}
\newtheorem{rmk}{Remark}[section]

\theoremstyle{definition}
\newtheorem{algorithm}[Theorem]{Algorithm}

\newcommand{\Real}[1][]{\ensuremath{\mathbb{R}^{#1}}}
\newcommand{\IP}[2]{\ensuremath{\left< #1, #2 \right>}}
\newcommand{\Norm}[1]{\ensuremath{\left\| #1 \right\| }}
\newcommand{\Abs}[1]{\ensuremath{\left| #1 \right|}}
\newcommand{\Deriv}[2][]{\ensuremath{\frac{\diff #1}{\diff #2}}}
\newcommand{\PDeriv}[2][]{\ensuremath{\frac{\partial #1}{\partial #2}}}
\newcommand{\diff}{\, {\rm d}}
\newcommand{\Diff}{\mathbf{D}}

\newcommand{\yesnumber}{\addtocounter{equation}{1}\tag{\theequation}}

\newcommand{\G}{\mathbf{G}}
\newcommand{\M}{\mathbf{M}}
\newcommand{\I}{\mathbf{I}}
\newcommand{\R}{\mathbf{R}}
\newcommand{\btheta}{{\bm{\theta}}}
\newcommand{\p}{\mathbf{p}}

\newcommand{\tilU}{\tilde{U}}
\newcommand{\tilbtheta}{{\tilde{\btheta}}}
\newcommand{\tilp}{\tilde{\p}}
\renewcommand{\v}{\mathbf{v}}
\renewcommand{\u}{\mathbf{u}}

\newcommand{\F}{\mathbf{F}}
\newcommand{\Fe}{\mathbf{F}_{\epsilon}}
\newcommand{\e}{\mathbf{e}}
\newcommand{\bGamma}{\bm{\Gamma}}
\newcommand{\g}{\bm{g}}	
\newcommand{\gpar}{g_{\parallel}}
\newcommand{\gort}{g_{\perp}}

\newcommand{\blind}{1}

\begin{document}

\def\spacingset#1{\renewcommand{\baselinestretch}%
{#1}\small\normalsize} \spacingset{1}


\if1\blind
{
  \title{\bf Geometrically Tempered Hamiltonian Monte Carlo}
  \author{Akihiko Nishimura \\
    Department of Mathematics, Duke University \\
    and \\
    David B.\ Dunson \\
    Department of Statistical Science, Duke University}
  \maketitle
} \fi

\if0\blind
{
  \bigskip
  \bigskip
  \bigskip
  \begin{center}
  \spacingset{1.8}
    {\LARGE \bf Geometrically Tempered Hamiltonian Monte Carlo}
  \end{center}
  \medskip
} \fi

\bigskip
\begin{abstract}
Hamiltonian Monte Carlo (HMC) has become routinely used for sampling from posterior distributions. Its extension Riemann manifold HMC (RMHMC) modifies the proposal kernel through distortion of local distances by a Riemannian metric.  The performance depends critically on the choice of metric, with the Fisher information providing the standard choice.  In this article, we propose a new class of metrics aimed at improving HMC's performance on multi-modal target distributions. We refer to the proposed approach as geometrically tempered HMC (GTHMC) due to its connection to other tempering methods.  We establish a geometric theory behind RMHMC to motivate GTHMC and characterize its theoretical properties. Moreover, we develop a novel variable step size integrator for simulating Hamiltonian dynamics to improve on the usual St\"{o}rmer-Verlet integrator which suffers from numerical instability in GTHMC settings.  We illustrate GTHMC through simulations, demonstrating generality and substantial gains over standard HMC implementations in terms of effective sample sizes. 
\end{abstract}

\noindent%
\vfill

\newpage
\spacingset{1.45} 
\section{Introduction}
	Markov chain Monte Carlo (MCMC) is routinely used for Bayesian inference to generate samples from posterior distributions. Metropolis-Hastings (MH) provides a general subclass of algorithms adaptable to a broad range of posterior distributions, without the need for special structures such as conjugacy. Many MH algorithms are highly inefficient, however, and Hamiltonian Monte Carlo (HMC) has emerged as one of the most reliable approaches for efficient sampling in general settings.  The STAN software package takes advantage of this generality and performance \citep{stan16}.  It is well known, however, that HMC faces major problems when posterior distributions are multimodal.  This article attempts to address this problem to obtain a general approach for accelerating mixing of HMC including in multimodal cases.
	
	Hamiltonian dynamics generates trajectories that move along the level sets of a scalar function commonly referred to as a Hamiltonian or energy. This property is known as \textit{conservation of energy} in physics. HMC exploits this property to generate proposals that are far away from the current state yet are accepted with high probability. If the parameter of interest has a distribution with multiple modes separated by a region of low probability density, however, the conservation of energy almost completely eliminates the possibility of HMC transitioning from one mode to another in a small number of iterations (cf. Section~\ref{subsec:HMC_multi_modality} or \citet{neal10}). This issue is inherent in the choice of Hamiltonian dynamics underlying HMC's proposal mechanism and consequently most variations of HMC \citep{hoffman14, neal94, shahbaba13, sd14} similarly suffer in the presence of multi-modality.
	
	\cite{girolami11} proposed Riemann manifold HMC (RMHMC), an extension of HMC that modifies the underlying Hamiltonian dynamics through distortion of local distances by a Riemannian metric. Their choice of metric, Fisher information, is not designed to facilitate sampling from a multimodal target distribution, but their work spurred a question: can a metric be chosen to help HMC sample more efficiently from multi-modal distributions? (See ``Discussion on the paper'' section in \cite{girolami11}.)  In this paper, we provide a positive answer to this question by proposing a class of metric specifically designed to lower the ``energy barriers'' among the modes, thereby enabling trajectories of Hamiltonian dynamics to transition from one mode to another more frequently. We call RMHMC under this class of metric as geometrically tempered HMC (GTHMC) due to its similarities to other tempering methods. While geometric methods in statistics are usually motivated using the language of intrinsic geometry \citep{amari00, girolami11, xifara14}, we develop a geometric theory behind RMHMC using the language of extrinsic geometry, thereby making the results more explicit and intuitive as well as accessible to a wider audience.
	
	Choosing a metric to adapt HMC to multimodal target distributions was previously considered by \cite{lan14}. Their approach, however, requires knowledge of the mode locations, substantial hand tuning, and ad hoc additions of drifts to the dynamics which can in general undermine the desirable properties of RMHMC. Many of these issues arise from the lack of precise treatment of geometry behind RMHMC and are all solved by GTHMC. Another related work is \cite{roberts03} where they consider using what they call Langevin tempered dynamics as a proposal generation mechanism. This dynamics is a Langevin dynamics analogue of Hamiltonian dynamics under isometric tempering, a special case of our geometric tempering method discussed in Section~\ref{sec:ITHMC_and_DTHMC}. Both Langevin and Hamiltonian dynamics explore the parameter space with highly variable velocities under geometric tempering, making their discrete approximation challenging (Section~\ref{sec:variable_step_integrator} and \cite{roberts03}). The deterministic nature of Hamiltonian dynamics, however, allows an accurate approximation of the dynamics in a relatively efficient manner through the variable step size integrator proposed in Section~\ref{sec:variable_step_integrator}.
	
	The rest of the paper is organized as follows. In Section~\ref{sec:GTHMC_and_motivation}, we motivate our choice of metric for GTHMC by developing geometric intuitions behind RMHMC using the language of extrinsic geometry. Section~\ref{sec:ITHMC_and_DTHMC} provides two example classes of GTHMC algorithms. Section~\ref{sec:variable_step_integrator} develops a novel variable stepsize integrator for Hamiltonian dynamics, motivated by the need for an improvement over the standard St\"{o}rmer-Verlet scheme that produces unstable trajectories in GTHMC settings. An effective application of the variable stepsize integrator to GTHMC and other HMC variants calls for an improved acceptance-rejection mechanism, and this is also described in Section~\ref{sec:variable_step_integrator}. Both the integrator and acceptance-rejection algorithm are general tools of independent interest. In Section~\ref{sec:simulation}, we compare the performance of GTHMC to HMC on various examples and demonstrate its superiority in terms of effective sample sizes. 

\section{Motivation and geometric theory behind GTHMC}
\label{sec:GTHMC_and_motivation}
	We begin this section with a brief review of RMHMC. We then discuss why HMC variants in general perform poorly on multimodal target distributions, which leads to a simple motivation for GTHMC defined in Section~\ref{sec:GTHMC_simple_motivation}. In the subsequent subsections, we develop more precise geometric theory behind GTHMC.

\subsection{Hamiltonian dynamics and RMHMC}
	Given a parameter of interest $\btheta$ with an (unnormalized) probability distribution $\pi(\btheta)$, RMHMC defines an augmented target distribution 
	\begin{equation} \label{eq:rmhmc_target}
	\pi(\btheta, \p)
		\propto \pi(\btheta) \thinspace \mathcal{N}(\p; \mathbf{0}, \G(\btheta))
	\end{equation}
	where $\mathcal{N}(\p; \mathbf{0}, \G(\btheta))$ denotes a probability density function of a centred multivariate Gaussian with a covariance matrix $\G(\btheta)$, sometimes called a \textit{mass matrix}. The corresponding Hamiltonian is defined as the negative logarithm of the joint target distribution (up to an additive constant):
	\begin{equation} \label{eq:rmhmc_hamiltonian}
	H(\btheta,\p) = - \log \pi(\btheta) + \log \Abs{\G(\btheta)}^{1/2} + \frac12 \p^T \G^{-1}(\btheta) \p
	\end{equation}
	RMHMC generates a proposal by simulating \textit{Hamiltonian dynamics}, where the evolution of the state $(\btheta, \p)$ is governed by a differential equation known as \textit{Hamilton's equations}:
	\begin{equation} \label{eq:Hamilton}
	\begin{aligned}
	\Deriv[\btheta]{t}
		&= \nabla_{\p} H(\btheta, \p), \quad
	\Deriv[\p]{t}
		= - \nabla_{\btheta} H(\btheta, \p)
	\end{aligned}
	\end{equation}
	Let $\F_\tau$ denote a solution operator of \eqref{eq:Hamilton} i.e.\ $\F_\tau(\btheta, \p) = (\btheta(\tau), \p(\tau))$ where $\left\{ \big(\btheta(t),\p(t) \big) \right\}_{t=0}^\tau$ is the solution of \eqref{eq:Hamilton} with the initial condition $(\btheta(0), \p(0)) = (\btheta, \p)$. Hamiltonian dynamics and its solution operator is \textit{reversible} in the sense that $(\R \circ \F_\tau)^{-1} = \R \circ \F_\tau$ where $\R$ is a momentum flip operator $\R(\btheta, \p) = (\btheta, - \p)$. Hamiltonian dynamics also satisfies the \textit{conservation of energy} property: $H(\btheta(t), \p(t)) = H(\btheta_0, \p_0)$ for all $t \in \Real$
	
	For the purpose of our discussion here, let us suppose that Hamilton's equation as above can be solved exactly. In this idealized situation and in its basic form, RMHMC works as follows. In the algorithm description below, a \textit{path length} $\tau$ is a fixed tuning parameter.
	
	\begin{algorithm}[RMHMC w/o numerical approximation]
	\label{alg:exact_rmhmc}
	RMHMC samples from the distribution \eqref{eq:rmhmc_target} by repeating the following steps:
	\begin{enumerate}[1)]
	\item Given the current state $(\btheta, \p)$, re-sample the momentum $\p | \btheta \sim \mathcal{N}(\p; \mathbf{0}, \G(\btheta))$.
	\item Propose $(\btheta^*, \p^*) = (\btheta(\tau),-\p(\tau))$, where $\left\{ \big(\btheta(t),\p(t) \big) \right\}_{t=0}^\tau$ is the solution of \eqref{eq:Hamilton} with the initial condition $(\btheta(0), \p(0)) = (\btheta, \p)$ and a Hamiltonian defined as in \eqref{eq:rmhmc_hamiltonian}. 
	\item Accept the proposal with probability 1. (The acceptance probability is 1 in the absence of the approximation error in a numerical solution \eqref{eq:Hamilton}.) 
	\end{enumerate}
	\end{algorithm}
	RMHMC recovers a familiar HMC \citep{duane87, neal10} when the mass matrix is independent of the position variable. See \cite{neal10} and \cite{girolami11} for more detailed presentations on HMC and RMHMC.

\newcommand{\m}{\scalebox{0.75}[1.0]{\(-\)}}
\subsection{Multi-modality and Conservation of Energy}
\label{subsec:HMC_multi_modality}
	We now explain how existing HMC variants suffer from multimodality in the target distribution. For simplicity we consider the basic version of (Riemann manifold) HMC as in Algorithm~\ref{alg:exact_rmhmc} with a constant mass matrix $\M$, but the following analysis applies equally to the other HMC variants. 
	
	It is useful to consider a Hamiltonian as a sum of potential energy $U(\btheta) = - \log \pi(\btheta)$ and kinetic energy $K(\p) = \frac12 \p^T \M^{-1} \p$ (up to an additive constant). The \textit{energy barrier} with respect to a potential energy function $U$ from a position $\btheta_1$ to $\btheta_2$ is the smallest possible energy increase along a continuous path from $\btheta_1$ to $\btheta_2$:
	\begin{equation} \label{eq:energy_barrier}
	B(\btheta_1, \btheta_2; U) :=
		\inf_{\bm{\gamma} \in C^0} \Big\{ \max_{0 \leq t \leq 1} U(\bm{\gamma}(t)) - U(\btheta_1) \ \Big| \thinspace \bm{\gamma}(0) = \btheta_1 \text{ and } \bm{\gamma}(1) = \btheta_2 \Big\}
	\end{equation}	 
	where $C^0$ denotes a class of continuous functions. The quantity $B(\btheta_1, \btheta_2 ; U)$ is the minimum amount of kinetic energy from Step 1 of Algorithm~\ref{alg:exact_rmhmc} needed for HMC to reach $\btheta_2$ from $\btheta_1$ in a single iteration; to see this, notice that a trajectory of Hamiltonian dynamics satisfies the following relation due to the conservation of energy:
	\begin{equation} \label{eq:energy_transfer}
	U(\btheta(t)) - U(\btheta_0)
		=  K(\p_0) - K(\p(t))
		\leq K(\p_0)
	\end{equation}
	where $\btheta(0) = \btheta_0$ and $\p(0) = \p_0$. The quantity $K(\p_0) - K(\p(t))$ is the amount of energy transferred from kinetic to potential at time $t$. Since the increase in the potential energy along a trajectory is upper bounded by $K(\p_0)$, the trajectory generated in Step 2 of HMC will not be able to reach $\btheta_2$ if the kinetic energy from Step 1 is smaller than $B(\btheta_1, \btheta_2 ; U)$. This is problematic for HMC as the energy barrier $B(\btheta_1, \btheta_2 ; U)$ would be high if $\btheta_1$ and $\btheta_2$ were two modes of $\pi(\btheta)$ with a region of low probability in between. To make things worse, there is no guarantee that a momentum variable with minimum required kinetic energy will actually generate a path between two modes. 
	
	
\subsection{Simple Motivation for GTHMC}
\label{sec:GTHMC_simple_motivation}
	We now define GTHMC and provide a simple motivation behind it.
	\begin{Defn}[GTHMC] \label{def:gthmc}
	GTHMC is a sub-class of RMHMC in which a metric $\G_T(\btheta)$ satisfies the following relation for $T > 1$
		\begin{equation} \label{eq:gthmc_def}
		\Abs{\G_T(\btheta)}^{1/2} \propto \pi(\btheta)^{1 - \frac{1}{T}}
		\end{equation}
	where $\Abs{\M}$ denotes the determinant of a matrix $\M$.
	\end{Defn}
	
	\noindent For GTHMC at temperature $T$, the Hamiltonian decomposes into a potential energy $U_T(\btheta)$ and kinetic energy $K_T(\btheta,\p)$ where
	\[ U_T(\btheta) = - \frac{1}{T} \log \pi(\btheta),
	\quad K_T(\btheta,\p) = \frac12 \p^T \G_T^{-1}(\btheta) \p \]
	and $K_T(\btheta, \p) \sim \chi_d^2 / 2 $ irrespective of $T$. As in the HMC setting \eqref{eq:energy_transfer}, the conservation of energy implies that
	\[ U_T(\btheta(t)) - U_T(\btheta_0)
		= K_T(\btheta_0,\p_0) - K_T(\btheta(t),\p(t)) 
	 \]
	where $(\btheta(t), \p(t))$ denotes a trajectory of the corresponding Hamiltonian dynamics with the initial condition $(\btheta_0, \p_0)$. Again, $K_T(\btheta_0,\p_0)$ is the maximum possible increase in the potential energy along the trajectory of Hamiltonian dynamics. Now notice that  $U_T(\btheta) = U_1(\btheta) / T$ and therefore we have
	\[ B(\btheta_1, \btheta_2 ; U_T)
		= \frac{1}{T} \thinspace B(\btheta_1, \btheta_2 ; U_1) \] 
	Hence, the energy barrier from $\btheta_1$ to $\btheta_2$ becomes lower as $T$ becomes large, requiring less kinetic energy for the trajectory to reach $\btheta_2$ from $\btheta_1$. As we discuss in more detail below (in particular, see the remark in Section~\ref{sec:theory_behind_GT}), a property similar to \eqref{eq:gthmc_def} is not only a convenient way but also a requirement to allow RMHMC to move from one mode to another in a small number of iterations.

\subsection{Geometric Intuition behind RMHMC}
\label{subsec:geom_intuition_behind_RMHMC}
	To further motivate GTHMC, we establish a theoretical result on RMHMC that provides a novel geometric intuition behind the algorithm. Our approach is to describe the Hamiltonian dynamics underlying RMHMC in terms of a more intuitive Newtonian dynamics on a manifold embedded in a Euclidean space. The required knowledge of Riemannian geometry is minimal and the supplemental appendix provides further background information.
	
	
	
\subsubsection{Newtonian Dynamics on a Manifold}
	We first review Newtonian dynamics on a Euclidean space, which can be considered as a special case of Hamiltonian dynamics when the mass matrix is proportional to the identity and the Hamiltonian takes the form $H(\btheta,\p) = U(\btheta) + \p^T \p / 2$ for a potential energy function $U(\btheta)$.
 In this case, the Hamilton's equation recovers Newtonian mechanics' description of the motion of a unit-mass particle in the potential energy field $U(\btheta)$: 
	\begin{align*}
	\Deriv[\btheta]{t}
		= \p, \quad
	\Deriv[\p]{t}
		= - \nabla U(\btheta)	
	\end{align*}
	The first equation simply expresses the fact that velocity is a time derivative of position. The second equation expresses Newton's second law; acceleration is proportional to force, the negative gradient of potential energy in our case. Borrowing from Neal (2010), Newtonian dynamics in two dimensions can be imagined as a motion of a frictionless puck that slides over a surface of height $U(\btheta)$. At the position $\btheta$, the puck experiences a force in the direction of greatest descent $-\nabla U(\btheta)$. If the surface is flat around $\btheta$ (i.e.\ $\nabla U \equiv \mathbf{0}$), the puck continues to move at a constant velocity in this area.

	Now consider a potential energy $\tilU(\tilbtheta)$ defined on a $d$-dimensional manifold $M \subset \Real[\tilde{d}]$ and let $T_{\tilbtheta} M \subset \Real[\tilde{d}]$ denote the tangent space of $M$ at $\tilbtheta$. As in a Euclidean space, Newtonian dynamics on a manifold describes the motion of a particle under the potential energy field $\tilU(\tilbtheta)$ driven in the direction of the greatest energy decrease, except that the particle is now constrained on a manifold $M$. Denoting the gradient of $\tilde{U}$ on $M$ by 
$\nabla^M \tilU(\tilbtheta)$, Newtonian dynamics on a manifold is defined as follows:
	\begin{Defn}[Newtonian dynamics on a manifold]
	\label{def:newtonian_dynamics_on_manifolds}
	A trajectory of Newtonian dynamics on a manifold $M$ under the potential energy field $\tilU(\tilbtheta)$ with an initial condition $\tilbtheta_0 \in M$ and $\tilp_0 \in T_{\tilbtheta_0} M$ is a unique solution $(\tilbtheta(t), \tilp(t)) \in M \times T_{\tilbtheta(t)} M$ of the differential equation
	\begin{align*}
	\Deriv[\tilbtheta]{t}
		= \tilp, \quad
	\Deriv[\tilp]{t}
		= - \nabla^M \tilU(\tilbtheta) 
	\end{align*}
	such that $\tilbtheta(0) = \tilbtheta_0$ and $\tilp(0) =  \tilp_0$. 
	\end{Defn}

\newcommand{\prodmap}{\g^{-1} \hspace{-.5ex} \times \hspace{-.2ex} \Diff \g^T}

\subsubsection{RMHMC in terms of Newtonian Dynamics on a Manifold}
\label{sec:RMHMC_as_Newtonain_dynamics_on_manif}
	Just as in Euclidean space, we can run HMC on a manifold by solving the Newtonian dynamics to generate samples from a given target distribution. We introduce HMC on a manifold as a theoretical tool to enhance our understanding of RMHMC, so we do not concern ourselves with how the Newtonian dynamics on a manifold may be numerically approximated. 
	
	\begin{Defn}[HMC on a manifold]
	\label{def:HMC_on_manif}
	 Given a pdf $\tilde{\pi}(\tilbtheta)$ on a manifold $M$, the following procedures generate a Markov chain $\{(\tilbtheta^{(i)}, \tilp^{(i)}) \}_{i=1}^{\infty}$ whose stationary distribution has the marginal $\tilde{\pi}(\tilbtheta)$. 1.\ Sample $\tilp^{(i)}$ from the standard Gaussian on $T_{\tilbtheta^{(i)}}M$. 2.\ Set $(\tilbtheta^{(i+1)}, \tilp^{(i+1)}) = (\tilbtheta^{(i)}(\tau), -\tilp^{(i)}(\tau))$ where $\{(\tilbtheta^{(i)}(t), \tilp^{(i)}(t))\}_t$ is a solution of the Newtonian dynamics as in Definition~\ref{def:newtonian_dynamics_on_manifolds} with the potential energy $\tilde{U}(\tilbtheta) = - \log \tilde{\pi}(\tilbtheta)$ and the initial condition $(\tilbtheta^{(i)}, \tilp^{(i)})$.	
	\end{Defn}


\newcommand{\bSigma}{\bm{\Sigma}}


We now state our main theoretical result. Theorem~\ref{thm:RMHMC_as_reparam_HMC} below provides valuable insights into the behaviors of RMHMC trajectories that are hard to predict otherwise.

	
	\begin{Theorem}[RMHMC as reparametrization of HMC]
	\label{thm:RMHMC_as_reparam_HMC}
	Given a pdf $\pi(\btheta)$ on $\Real[d]$, let $\tilde{\pi}$ denote the pdf of a random variable $\g(\bm{\Theta})$ for $\bm{\Theta} \sim \pi$. For the initial input $\btheta_0 \in \Real[d]$ and $\tilbtheta_0 = \g(\btheta_0)$, let $\big\{(\tilbtheta^{(i)}, \tilp^{(i)})\big\}_{i=0}^N$ be a Markov chain generated by HMC on a manifold as in Definition~\ref{def:HMC_on_manif}. Then a Markov chain $\big\{\g^{-1} \hspace{-.5ex} \times \hspace{-.2ex} \Diff \g^T (\tilbtheta^{(i)}, \tilp^{(i)})\big\}_{i=0}^N$ on $\Real[d] \times \Real[d]$ defined through the map
	\begin{equation*}
	\prodmap (\tilbtheta, \tilp)
		:= \left( \g^{-1}(\tilbtheta), \, \Diff\g_{\g^{-1}(\tilbtheta)}^T \tilp \right)
	\end{equation*}
	has the same distribution as the Markov chain generated by running RMHMC on $\Real[d]$ with a metric $\G(\btheta) = \Diff\g_{\btheta}^T \Diff\g_{\btheta}$.
	\end{Theorem} 
	Theorem~\ref{thm:RMHMC_as_reparam_HMC} is a consequence of the fact that $\prodmap$ bijectively maps Newtonian dynamics on a manifold onto the corresponding Hamiltonian dynamics on $\Real[d]$, formally stated as follows:
	
	\begin{Theorem}
	\label{thm:Newton_Hamilton_bijection}
	If $(\tilde{\btheta}(t), \tilde{\p}(t))$ is a solution of the Newtonian dynamics on $M$ with a potential energy $\tilU(\tilde{\btheta}) = -\log \tilde{\pi}(\tilde{\btheta})$, then $(\btheta(t),\p(t)) = \prodmap (\tilde{\btheta}(t), \tilde{\p}(t))$ is a solution of Hamiltonian dynamics in $\Real[d]$ corresponding to the Hamiltonian \eqref{eq:rmhmc_hamiltonian} with $\G(\btheta) = \Diff\g_{\btheta}^T \Diff\g_{\btheta}$.
	\end{Theorem}
	
	Theorems~\ref{thm:RMHMC_as_reparam_HMC} and \ref{thm:Newton_Hamilton_bijection} in essence state that, up to numerical approximation errors in simulating trajectories, running RMHMC to sample from a parameter space $\btheta \in \Real[d]$ is equivalent to running HMC to sample from the reparametrization $\tilbtheta = \g(\btheta) \in M$. This means that the metric $\G(\btheta)$ should be chosen so that the reparametrization defines a well-conditioned distribution from which HMC can sample efficiently. In the special case when $\g = \widehat{\bSigma}^{-1/2}$ is a linear operator and $\G = \widehat{\bSigma}^{-1}$, Theorem~\ref{thm:RMHMC_as_reparam_HMC} recovers a well-known fact on the effect of using a non-identity mass matrix in HMC \citep{neal10}. The Langevin dynamics analogue of Theorem~\ref{thm:RMHMC_as_reparam_HMC} can also be established: see Supplement Section~\ref{app:geometric_theory_for_mmala}.
	
\begin{rmk}
	Theorem~\ref{thm:RMHMC_as_reparam_HMC} and \ref{thm:Newton_Hamilton_bijection} start with a reparametrization $\g$ and identify the corresponding Riemannian metric as $\G(\btheta) = \Diff \g_\btheta^T \Diff \g_\btheta$. Nash embedding theorem \citep{nash54} tells us that the construction can go in the other direction as well; given a metric $\G(\btheta)$, there is a corresponding (local) reparametrization $\g$ so that HMC in the space $\g(\btheta)$ is equivalent to RMHMC in the space $\btheta$ with a metric $\G(\btheta)$. 
\end{rmk}
	

\subsection{Theory behind geometric tempering}
\label{sec:theory_behind_GT}	
	Tempering methods are motivated by  the fact that a distribution $\pi(\btheta)^{1/T} / Z_T$, where $Z_T$ is a normalizing constant, has less severe multi-modality than $\pi(\btheta)$ for $T > 1$ \citep{earl05, geyer95, marinari92, neal01}. The main challenge is to relate the samples from the tempered distribution $\pi(\btheta)^{1/T} / Z_T$ back to the original target $\pi(\btheta)$. 
	
	GTHMC works by implicitly sampling from a transformed variable $\tilbtheta$ such that the transformation $\btheta \to \tilbtheta$ alleviates the multi-modality of $\pi(\btheta)$. This implicit transformation is achieved by appropriately modifying Hamiltonian dynamics for the parameter $\btheta$, and therefore the samples generated from GTHMC retain the target distribution $\pi(\btheta)$. Theorem~\ref{thm:RMHMC_as_reparam_HMC} implies that the use of a metric with a property $\Abs{\G_T(\btheta)}^{1/2} = \pi(\btheta)^{(1-\frac{1}{T})}$ corresponds to an implicit transformation $\tilbtheta = \g(\btheta)$ through a map $\g$ such that $\Abs{\Diff\g_{\btheta}^T \Diff\g_{\btheta}}^{1/2} = \pi(\btheta)^{(1-\frac{1}{T})}$. This means that the transformed variable $\tilbtheta$ would have the distribution
	\[ \tilde{\pi}(\tilbtheta) 
		\propto \pi \circ g^{-1}(\tilbtheta)^{1/T} \]
	by virtue of the (generalized) change of variable formula $\tilde{\pi}\big( \tilde{\btheta} \big) = \left| \Diff \g_\btheta^T \Diff \g_\btheta \right|^{- 1/2} \pi(\btheta)$ \citep{federer69}. This is how GTHMC effectively lowers the energy barriers among the modes of $\pi(\btheta)$ by a factor of $1/T$. Geometric tempering does not compete with existing tempering methods and in fact can be combined with them; see Section~\ref{sec:discussion}.
	
	\begin{rmk}
	The implicit reparametrization $\g: \btheta \to \tilbtheta$ under RMHMC has no effects on energy barriers among the modes if a metric $\G(\btheta)$ has a constant volume factor \mbox{$\Abs{\G(\btheta)}^{1/2} = c$}. More generally, the difference in the potential energy $\tilde{U}$ between two positions $\tilbtheta_1$ and $\tilbtheta_2$ is given by
	\[ \log \frac{\tilde{\pi}(\tilbtheta_2)}{\tilde{\pi}(\tilbtheta_1)} 
 		= \log \frac{\pi(\btheta_2)}{\pi(\btheta_1)} - \frac12 \log  \frac{\Abs{\G(\btheta_2)}}{\Abs{\G(\btheta_1)}} \]
 	where $\btheta_i = \g^{-1}(\tilbtheta_i)$. The above equation shows that RMHMC has a measurable effect on the energy difference between $\tilbtheta_1$ and $\tilbtheta_2$ only if
 	\[ \frac{\Abs{\G(\btheta_2)}}{\Abs{\G(\btheta_1)}} 
 		\underset{\sim}{\propto} \left( \frac{\pi(\btheta_2)}{\pi(\btheta_1)} \right)^\alpha  
 		\quad \text{ for } \alpha > 0. \]
Thus any metric designed to promote the movements among the modes must locally have a property like $\Abs{\G(\btheta)}^{1/2} \propto \pi(\btheta)^{(1-\frac{1}{T})}$ for $T > 1$, the defining characteristic of GTHMC.
	\end{rmk}

\section{Concrete examples of GTHMC}
\label{sec:ITHMC_and_DTHMC}
	
	We have only assumed $\Abs{\G(\btheta)}^{1/2} \propto \pi(\btheta)^{(1-\frac{1}{T})}$ in our development of GTHMC, leaving substantial flexibility in the choice of metric. We propose two computationally convenient variants of GTHMC, illustrating how the choice of metric affects performance.  Simulation results are presented in  Section~\ref{sec:simulation}, preceded by discussion on how to efficiently approximate the dynamics underlying GTHMC in Section~\ref{sec:variable_step_integrator}.	
	
\subsection{Isometrically tempered HMC (ITHMC)}
\label{sec:ithmc}
	The choice $\G(\btheta) = g(\btheta) \, \I$ with $g(\btheta) \propto \pi(\btheta)^{\frac{2}{d} (1-\frac{1}{T})}$ is arguably the simplest way to satisfy the requirement $\Abs{\G(\btheta)}^{1/2} \propto
	 \pi(\btheta)^{(1-\frac{1}{T})}$. This metric modifies the local distance of the parameter space uniformly in all the directions, and therefore we call GTHMC with this choice of metric \textit{isometrically tempered HMC} (ITHMC). 
	 
\subsection{Directionally tempered HMC (DTHMC)}
\label{sec:dthmc}
	As mentioned in Section~\ref{subsec:HMC_multi_modality}, an iteration of HMC with sufficient kinetic energy to overcome energy barriers does not guarantee a transition from one mode to another. The transition can be infrequent even for GTHMC when a randomly generated trajectory from one mode tends not to travel in the direction of another. For this reason, a non-uniform distortion of the local distances can improve efficiency in certain situations (see Section~\ref{sec:GTHMC_illustration}).  We let \textit{directionally tempered HMC} (DTHMC) refer to a version of GTHMC in which the local distance in a particular direction is modified differently from the other directions.  More precisely, we set 
	\begin{align} \label{eq:dthmc_metric}
	&\G(\btheta) = \gpar(\btheta) \, \u \u^T + \gort(\btheta) \, (\I - \u \u^T) 
	\end{align}
where $\gpar(\btheta) = \pi(\btheta)^{2 \gamma (1-\frac{1}{T})}$, 
$\gort(\btheta) = \pi(\btheta)^{2 \frac{1-\gamma}{d-1} (1-\frac{1}{T})}$,
 and $d^{-1} < \gamma \leq 1$. This metric modifies the distance only in the direction of $\u$ when $\gamma = 1$ while it coincides with the metric for ITHMC when $\gamma = d^{-1}$. This kind of metric is appropriate when it is known that the multi-modality is more severe in a particular direction.

	The metric proposed in \cite{lan14} has an  apparent similarity to \eqref{eq:dthmc_metric} but lacks the crucial property $\Abs{\G(\btheta)}^{1/2} \propto \pi(\btheta) ^{(1-\frac{1}{T})}$. Some degree of geometric tempering is achieved in their examples as a result of substantial manual tuning of the metrics based on the knowledge of mode locations. Even then, they have to resort to ad hoc additions of drifts to their dynamics to induce more frequent transitions among the modes. GTHMC provides more effective geometric tempering without such an extensive manual tuning.
	
 	
\subsection{Illustration of trajectories generated by GTHMC}
\label{sec:GTHMC_illustration}
	We simulate some trajectories of HMC, ITHMC and DTHMC to illustrate the effect of geometric tempering as well as the difference between isometric and directional tempering. We construct a bi-modal target distribution $\pi(\cdot)$ as a mixture of 2-d standard Gaussians centered at $(-4,0)$ and $(4,0)$. For each of the algorithms, trajectories are simulated for $t = 3$ from a high density region near $(-4,0)$, all having the same initial kinetic energy $K(\btheta_0,\p_0) = \p_0^T \G(\btheta_0)^{-1} \p_0 \, / \, 2 = 0.8$. For DTHMC, the tempering direction is along the $x$-axis (i.e.\ $\u = (1,0)$ in Equation~\eqref{eq:dthmc_metric}) and the temperature is set at $T = 15$ for both DTHMC and ITHMC. Between the two modes, the energy barrier with respect to the potential energy $U = - \log(\pi)$ is roughly given by $U(0,0) - U(-4,0) \approx 7.3$, so the geometrically tempered trajectories have more than enough kinetic energy to overcome the barrier as $K(\btheta_0,\p_0) = 0.8 > 7.3 \, / \, T$.
	
	Figure~\ref{fig:comparison_bet_HMC_DTHMC_ITHMC} shows the trajectories generated as described above. For HMC and DTHMC, the trajectories of the same color are meant to be directly comparable as they have exactly the same value of $(\btheta_0, \G(\btheta_0)^{-1/2} \p_0)$ (recall that $\G(\btheta)^{-1/2} \p | \btheta \sim \text{Normal}(\mathbf{0}, \I)$ irrespective of the choice of a metric). The ITHMC trajectories were given similar but not necessarily comparable values of $(\btheta_0, \G(\btheta_0)^{-1/2} \p_0)$; the initial conditions were instead chosen to better highlight the difference between the isometric and directional tempering. 
	
	As can be seen, none of the HMC trajectories have sufficient (total) energy to reach the other mode and consequently are trapped near the left mode. On the other hand, the DTHMC trajectories can easily reach the other mode with high probability. The ITHMC trajectories also have enough energy to travel through the low probability and clearly improve on HMC, but are not as successful as DTHMC in locating the other mode. In general, geometrically tempered trajectories tend to drift toward regions of lower probability as the distances to those regions are closer than to regions of higher probability under the metric of the form \eqref{eq:dthmc_metric}. Benefits of geometric tempering therefore are greater if done in particular directions of interest to limit the exploration of irrelevant regions.
	
	Along each of the trajectories, asterisk signs are placed at $\{\btheta(t_i)\}_{i=0}^n$, where $\{t_i\}_{i=0}^n$ partitions $[0, t]$ into $n$ equally spaced intervals.  This is done to demonstrate how the velocity of a trajectory changes along its path. The tempered trajectories travel through low probability density regions in a relatively small amount of time, a property we discuss further in Section~\ref{subsec:GTHMC_velocity}.
	
	The cyan coloured DTHMC trajectory deserves some attention. The large oscillation in the tempered direction can be understood as follows in view of Theorem~\ref{thm:RMHMC_as_reparam_HMC}: 
	a map $\g: \Real[2] \to M$ corresponding to the DTHMC metric heavily compresses the distance along the $x$-axis in low probability regions. Therefore, a small oscillation of a trajectory on $M$ manifests as a large oscillation in the original parameter space $\Real[2]$. This phenomenon does not negatively affect the mixing of DTHMC but it does increase the computational cost; see our simulation results in Section~\ref{sec:simulation}. 
	

	\begin{figure}[!htb]
	\begin{minipage}{1.06\textwidth}
	\hspace{-6ex}
    \centering
    	\subfigure[Trajectories of HMC]{
	        \includegraphics[width=.23\textwidth]{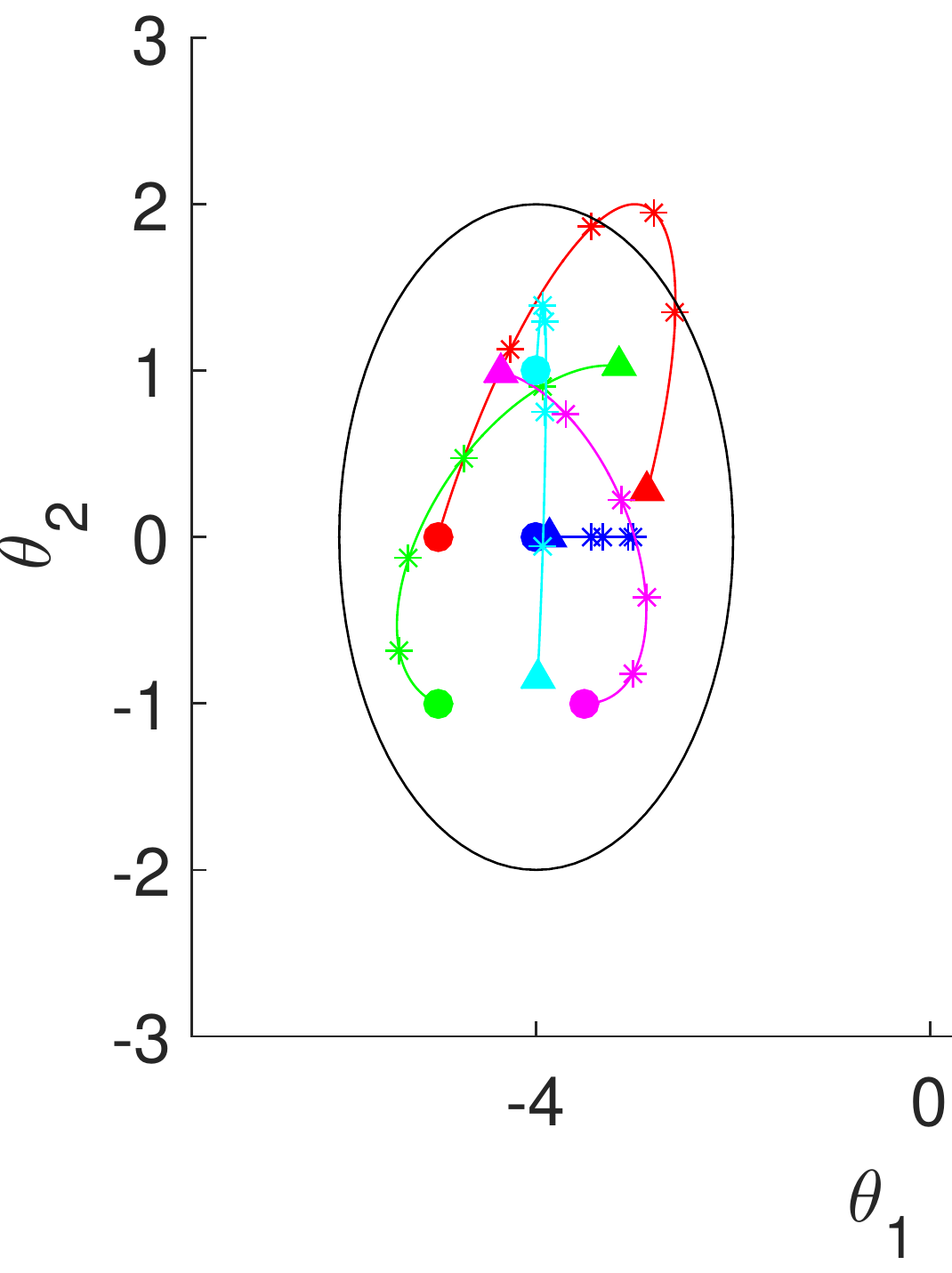} } \hspace{-3ex}
    	\subfigure[Trajectories of DTHMC]{
	        \includegraphics[width=.36\textwidth]{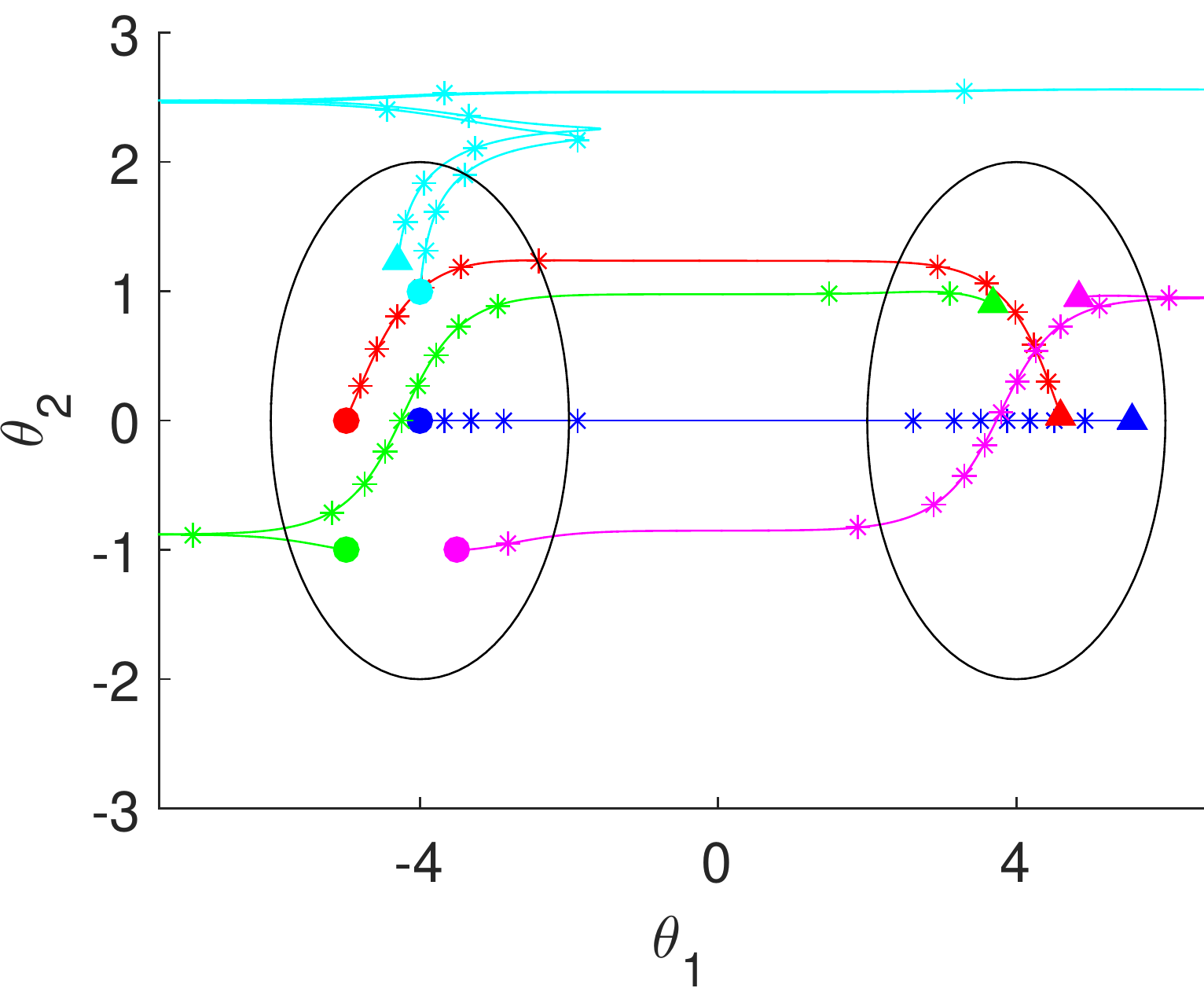} }
	  	\hspace{-2ex}
	    \subfigure[Trajectories of ITHMC]{
	        \includegraphics[width=.36\textwidth]{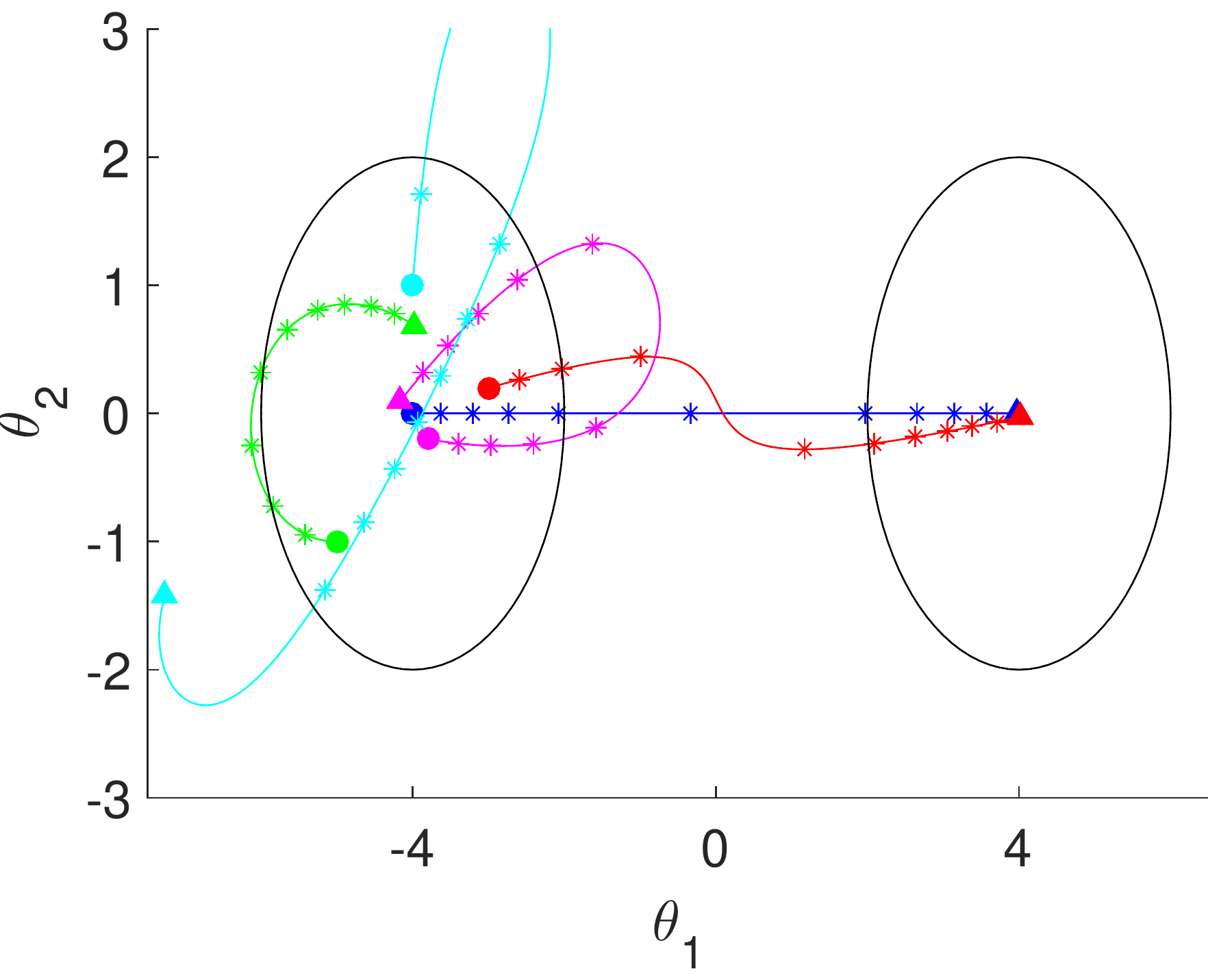} }
	\end{minipage}
	\caption{Comparison of trajectories generated (a) without tempering, (b) with directional, and (c) with isometric tempering. The black circles indicate a high probability density region. The circular and triangular markers indicate the start and end point of the trajectories. The star marks are placed at equal time intervals. (The time interval varies from plot to plot but is constant within each plot.)}
    \label{fig:comparison_bet_HMC_DTHMC_ITHMC}
	\end{figure}

\section{Reversible variable-step integrator for GTHMC}
\label{sec:variable_step_integrator}

Until this point, we have put aside the issue that Hamiltonian dynamics in general cannot be solved exactly. The usual St\"{o}rmer-Verlet scheme for approximating Hamiltonian dynamics encounters numerical stability issues in GTHMC. This is because the velocity $\diff \btheta / \diff t = \G_T^{-1}(\btheta)\p$ can become unboundedly large in regions of low probability.  We begin this section by quantifying this phenomena and follow it up with the development of a novel reversible integrator that overcomes this shortcoming of St\"{o}rmer-Verlet and enables practical applications of GTHMC. We then provide concrete examples of the integrator applied to ITHMC and DTHMC in Section~\ref{sec:integrator_for_ithmc_and_dthmc}.
	
\subsection{Velocity of GTHMC trajectories}
\label{subsec:GTHMC_velocity}
The velocity of a GTHMC trajectory grows rapidly as it enters a low probability region in which $\pi(\btheta) / \Norm{\pi}_\infty \ll 1$ where $\Norm{\pi}_\infty = \max_{\btheta} \pi(\btheta)$.
This is a necessary consequence of the fact that GTHMC travels through such regions without modifying the target distribution $\pi(\btheta)$; a dynamics would distort a distribution if it spends as much time in low probability regions as in high probability regions. The position coordinate of a GTHMC trajectory $\btheta(t)$ travels faster and faster as $\pi(\btheta(t))$ becomes smaller, thereby spending less time in regions with lower probability. While this enables GTHMC to transition from one mode to another, this property also makes it difficult to approximate GTHMC trajectories with a fixed step size integrator like St\"{o}rmer-Verlet. 
		
	To quantify how the velocity of a GTHMC trajectory depends on position, consider an exact (not numerically approximated) GTHMC trajectory $(\btheta(t),\p(t))$ with an initial condition $(\btheta_0, \p_0)$ drawn from the stationary distribution $\pi(\btheta,\p)$. The energy and volume conservation property of Hamiltonian dynamics implies $(\btheta(t),\p(t)) \overset{d}{=} (\btheta_0,\p_0)$ and therefore $\G_T^{-1/2}(\btheta(t)) \p(t) \sim \mathcal{N}(\mathbf{0}, \I)$ for all $t$. This suggests that the magnitude of the velocity $\diff \btheta / \diff t = \G_T^{-1}(\btheta)\p$ can grow as large as $\Norm{\G_T(\btheta(t))^{-1/2}}$ along a typical trajectory of GTHMC. Notice that, due to the constraint $\Abs{\G_T(\btheta)^{1/2}} \propto \pi(\btheta)^{(1-\frac{1}{T})}$, the matrix norm $\Norm{\G_T(\btheta)^{-1/2}}$ necessarily becomes unbounded as $\pi(\btheta) \to 0$ for $T > 1$.

\subsection{Explicit adaptive integrator with time rescaling}
\label{sec:explicit_adaptive_integrator}

The discussion in Section~\ref{subsec:GTHMC_velocity} suggests that GTHMC requires a \textit{variable stepsize} or \textit{adaptive} integrator that adjusts stepsize locally according to the current position. 
	Variable stepsize integrators can be interpreted as fixed stepsize integrators of a differential equation under \textit{time rescaling}. If $(\btheta(t),\p(t))$ denotes a solution of Hamilton's equations and a new time-scale $s$ is defined via the relation $\eta(\btheta) \diff s = \diff t$, the trajectory $(\btheta(s),\p(s))$ satisfies the following \textit{time rescaled Hamilton's equations}:
	\begin{equation} \label{eq:time_rescaled_Hamilton}
	\begin{aligned}
	\Deriv[\btheta]{s}
		&= \eta(\btheta) \nabla_{\p} H(\btheta, \p), \quad
	\Deriv[\p]{s}
		&= - \eta(\btheta) \nabla_{\btheta} H(\btheta, \p)
	\end{aligned}
	\end{equation}
	An implicit integrator similar to \textit{adaptive St\"{o}rmer-Verlet} of \cite{huang97} can be used to solve \eqref{eq:time_rescaled_Hamilton}
	The implicit updates of (adaptive) St\"{o}rmer-Verlet, however, require numerically solving for fixed points of non-linear functions and is a significant computational burden \citep{hairer06}.
	
	In order to address the above issues, we develop an explicit reversible integrator with built-in local stepsize adjustment. The integrator is a generalization of the one proposed by \cite{lan15} based on a similar variable transformation idea.
	In RMHMC settings, a Hamiltonian has the form $H(\btheta,\p) = \phi(\btheta) + \frac12 \p^T \G^{-1}(\btheta) \p$, and \eqref{eq:time_rescaled_Hamilton} can be written as:
	\begin{equation} \label{eq:time_rescaled_rmhmc}
	\begin{aligned}
	\Deriv[\btheta]{s}
		&= \eta \thinspace \G^{-1} \p, \quad
	\Deriv[\p]{s}
		&= - \eta \nabla_{\btheta} \phi + \frac12 \eta \thinspace \p^T \G^{-1} (\nabla_{\btheta} \G) \G^{-1} \p
	\end{aligned}
	\end{equation}
	where $\u^T (\nabla \M) \bm{w}$ denotes a vector whose $k$th entry is $\u^T (\partial_k \M) \bm{w}$ for $\u, \bm{w} \in \Real[d]$ and a $d \times d$ matrix valued function $\M$. With an appropriately chosen time-rescaling $\eta(\btheta)$, the differential equation \eqref{eq:time_rescaled_rmhmc} is much better-behaved than the equation in the original time scale. In fact, the choice $1/ \eta(\btheta) \propto \Norm{\G^{- 1/2}(\btheta)}$ stabilizes RMHMC trajectories in general as can be shown by an analysis similar to that of Section~\ref{subsec:GTHMC_velocity}. We now reparametrize the differential equation \eqref{eq:time_rescaled_rmhmc} in terms of the variables $(\btheta(s),\v(s))$ where $\v = \eta \thinspace \G^{-1} \p$. 
	After carrying out calculations described in Supplement Section~\ref{app:explicit_reversible_integrator}, we find that a trajectory $(\btheta(s),\v(s))$ satisfies the following differential equation:
	\begin{equation} \label{eq:vel_diff_eq_noncompact}
	\begin{aligned}
	\Deriv[\btheta]{s}
		&= \v, \quad
	\Deriv[v_k]{s}
		&= -\eta^2 \left[\G^{-1} \nabla_{\btheta} \phi \right]_k + \v^T \bGamma^k \v \ \text{ for } k = 1, \ldots, d 
	\end{aligned}
	\end{equation}
	where $[\mathbf{w}]_k$ denotes the $k$-th coordinate of $\mathbf{w}$ and $\bGamma^k = \bGamma^k(\btheta)$ denotes a symmetric matrix whose entries are defined as
	\begin{equation} \label{eq:bGamma_sym}
	\bGamma^k_{ij}
		= \sum_\ell (\G^{-1})_{k\ell} \left[ \frac12 \PDeriv{\theta_\ell} G_{ij} 
		- \frac{\eta}{2} \PDeriv{\theta_i} \left( \frac{1}{\eta} G_{\ell j} \right)
		- \frac{\eta}{2} \PDeriv{\theta_j} \left( \frac{1}{\eta} G_{\ell i} \right) \right]
	\end{equation}
	A reversible integrator of \eqref{eq:vel_diff_eq_noncompact} can be obtained by a symmetric linearly implicit scheme of Kahan \citep{lan15, ss94}, which results in the following update equations:
	\begin{equation} \label{eq:explicit_integrator_updates}
	\begin{aligned}
	\v_{1/2} 
		&= 
		\left( \I - \frac{\epsilon}{2} \v_0^T \bGamma(\btheta_0) \right)^{-1}
		\Big( \v_0 - \frac{\epsilon}{2}  \eta^2(\btheta_0) \G^{-1}(\btheta_0) \nabla \phi(\btheta_0) \Big) \\
	\btheta_{1}
		&= \btheta_0 + \epsilon \v_{1/2} \\
	\v_{1} 
		&=  
		\left( \I - \frac{\epsilon}{2} \v_{1/2}^T \bGamma(\btheta_1) \right)^{-1}
		\Big( \v_{1/2} - \frac{\epsilon}{2} \eta^2(\btheta_1) \G^{-1}(\btheta_1) \nabla \phi(\btheta_1) \Big)
	\end{aligned}
	\end{equation}
	where $\epsilon$ is a fixed step size and $\v^T \bGamma$ denotes a matrix whose $k$th row corresponds to $\v^T \bGamma^k$. The symmetry of the integrator implies that the local error is of order $O(\epsilon^3)$ i.e.
	\[ (\btheta_1, \v_1)(\epsilon) = \Fe(\btheta_0, \v_0) + O(\epsilon^3) \]
	where $\Fe$ is the solution operator of the dynamics \eqref{eq:vel_diff_eq_noncompact} \citep{leimkuhler04, neal10}. Unlike St\"{o}rmer-Verlet, this integrator is not volume-preserving, therefore the determinant of the Jacobian $\PDeriv[(\btheta_1, \v_1)]{(\btheta_0, \v_0)}$ needs to be included in the calculation of the acceptance probability in RMHMC applications (see Supplement Section~\ref{app:vlt_chmc}). We provide the derivation and further properties of the integrator in Supplement Section~\ref{app:explicit_reversible_integrator}. 	
	
\subsection{Examples: explicit adaptive integrator for ITHMC and DTHMC}
\label{sec:integrator_for_ithmc_and_dthmc}
	We illustrate how the time rescaling of Hamiltonian dynamics and resulting explicit integrator works in practice. With a metric defined as in Section~\ref{sec:ithmc} for ITHMC, we have $\Norm{\G^{-1/2}(\btheta)} = 1 / \sqrt{g(\btheta)}$, so we set $\eta(\btheta) = \sqrt{g(\btheta)}$. In this case, the matrix $\bGamma^k$ defined as \eqref{eq:bGamma_sym} becomes
	\begin{align}
	\bGamma^k 
		&= \frac{1}{2} \PDeriv[\log g]{\theta_k} \thinspace \I 
		- \frac14 \nabla_{\btheta} \log g \cdot \bm{e}_k^T 
		- \frac14 \bm{e}_k \cdot \nabla_{\btheta}^T \log g 
		\label{eq:ithmc_Gamma}
	\end{align}	 
	So we have
	\begin{align}
	\v^T \bGamma 
		&= \frac{1}{2} \nabla \log g \cdot \v^T 
		- \frac14 \IP{\v}{\nabla \log g} \I
		- \frac14 \v \cdot \nabla^T \log g 
		\label{ithmc_vGamma} \\
	\Longrightarrow 
	\left( \I - \frac{\epsilon}{2} \v^T \bGamma \right)
			&= \left(1 + \frac{\epsilon}{8} \IP{\v}{\nabla \log g} \right) \I
			- \frac{\epsilon}{4} \nabla_{} \log g \cdot \v^T
			+ \frac{\epsilon}{8} \v \cdot \nabla^T \log g \nonumber
	\end{align}
	Since the above matrix is a rank-2 perturbation of an identity, it can be inverted in $O(d)$ using the Sherman-Morrison formula to carry out the velocity updates in \eqref{eq:explicit_integrator_updates}:
	\begin{align*}
	\v^* 
		&= \left( \I - \frac{\epsilon}{2} \v^T \bGamma \right)^{-1} \left(\v + \frac{\epsilon}{2T} \nabla \log \pi \right)
	\end{align*}	
	The determinant $\Abs{\Diff \Fe}$ needed in the acceptance probability calculation can also be computed in $O(d)$ using the matrix determinant lemma (see \eqref{eq:v_Jacobian} in Supplement Section~\ref{app:explicit_reversible_integrator} for the formula of the Jacobian). 
	
	For DTHMC with a metric as in \eqref{eq:dthmc_metric}, we have $\Norm{\G^{-1/2}(\btheta)} = 1 / \sqrt{\gpar(\btheta)}$, so we set $\eta(\btheta) = \sqrt{\gpar(\btheta)}$. As in ITHMC, the numerical integration and determinant computation can be carried out in $O(d)$ because the matrix $\left( \I - \frac{\epsilon}{2} \v^T \bGamma \right)$ is a rank-3 perturbation of identity. The formulas for $\bGamma^k$ and $\v^T \bGamma$ are more complicated than those for ITHMC, however, and we refer the readers to the supplemental appendix for their full expressions.

\subsection{Variable length trajectory compressible HMC}

Although the variable step integrator of Section~\ref{sec:variable_step_integrator} enables an efficient and accurate approximation of otherwise unstable trajectories, the required time-rescaling of a Hamiltonian dynamics destroys its volume-preserving property. The modified acceptance-rejection scheme of compressible HMC (CHMC) \citep{fang14, lan15} can be used to preserve the stationary distribution, but its use in GTHMC settings generally suffers from low acceptance probabilities and poor mixing.
\if0\blind { To overcome this issue, we introduce \textit{variable length trajectory CHMC} (VLT-CHMC). By allowing individual trajectories to have different path lengths, VLT-CHMC constructs a transition kernel that better approximates the original dynamics and has a guaranteed high acceptance probability. A self-contained summary of the main idea and results on VLT-CHMC are provided in Supplement Section~\ref{app:vlt_chmc}. }\fi
\if1\blind{ Instead, we employ \textit{variable length trajectory CHMC} (VLT-CHMC) of \citet{nishimura16vlt}. By allowing individual trajectories to have different path lengths, VLT-CHMC constructs a transition kernel that better approximates the original dynamics and has a guaranteed high acceptance probability. A self-contained summary of the main idea and results on VLT-CHMC are provided in Supplement Section~\ref{app:vlt_chmc}.} \fi

\section{Simulations}
\label{sec:simulation}
	We compare the performance of HMC and GTHMC on various multi-modal target distributions to demonstrate the advantage of GTHMC. The effect of different temperatures and tempering schemes are also illustrated.
	
\subsection{Performance metric}
	Following \cite{hoffman14}, we compute the effective sample sizes (ESS) of marginal mean and variance estimators for each coordinate of a target distribution and report the minimum of these values. 
	For the majority of posterior distributions encountered in practice, the most computationally expensive parts of the algorithms are evaluations of $\nabla_{\btheta} \log \pi(\cdot)$. We therefore normalize ESS by the number of the gradient evaluations to account for the costs of each iteration. We also report ESS per 100 MCMC samples so that the qualities of the samples can be compared to independent ones. 
	
	The ESSs are estimated using the monotone sequence estimator of \cite{geyer92} with a small modification. In estimating the lag $k$ auto-covariance $a(k)$ of a statistic $g(\btheta)$, the true mean $\mu(g) := \mathbb{E}[g(\btheta)]$ is used in place of the empirical mean since this procedure leads to more reliable estimates of ESSs \citep{hoffman14}. The expectations were computed analytically or numerically with high accuracy. 

\subsection{Tuning parameters of HMC and GTHMC}
	Finding an optimal value of path length $\tau = \epsilon N$ for HMC is known to be difficult \citep{neal10}, so we use a variant of HMC known as the No-U-Turn-Sampler (NUTS) by \cite{hoffman14} which automatically adapts the path length for individual trajectories of Hamiltonian dynamics. The use of NUTS to benchmark against GTHMC is appropriate since NUTS uses the same underlying dynamics as HMC and has been shown empirically to perform as well as optimally tuned HMC in a variety of situations. The mass matrices of ITHMC and DTHMC as in Section~\ref{sec:ITHMC_and_DTHMC} degenerate to the identity when $T = 1$, so for fair comparison we used the identity mass matrix for NUTS. The stepsize $\epsilon$ was tuned using the dual-averaging algorithm of \cite{hoffman14} so that the average acceptance probability corresponds to a pre-specified value $\delta \in (0,1)$. Theoretical and empirical studies suggest the values of $\delta  \in [0.6, 0.8]$ to be optimal \citep{beskos13, hoffman14} and the values of $\delta = 0.5, 0.6, \ldots, 0.9$ were tried for each target distribution.
	
	For ITHMC and DTHMC, the parameters $\epsilon$ and $\tau$ were tuned alternately for a few times with one of them fixed while the other is adjusted. A modified dual-averaging algorithm was used to tune $\epsilon$ to achieve an appropriate acurracy in the numerical approximation of Hamiltonian dynamics. The path length $\tau$ was tuned to maximize a normalized expected squared jumping distance \citep{wang13}. 
	
	
\subsection{Example: bi-modal Gaussian mixture}
\label{sec:bimodal_example}
	We first compared the performance of NUTS and GTHMC on a simple bi-modal target distribution, a mixture of 2-d standard Gaussians centered at $(0,-4)$ and $(0,4)$ with equal weights as in Figure~\ref{fig:comparison_bet_HMC_DTHMC_ITHMC}. We ran ITHMC, DTHMC with $\gamma = .75$, and DTHMC with $\gamma = 1$ at different temperatures. DTHMC was tempered along the first coordinate. The performance of each algorithm is summarized in Table~\ref{table:ess_for_2d_bimodal}. ITHMC improves over NUTS substantially in terms of ESS, with further improvement obtained by DTHMC. Figure~\ref{fig:bimodal_traceplot} compares the traceplot of the best performing NUTS ($\delta = 0.7$) and DTHMC ($T = 20, \gamma = 1$). The efficiency gain by ITHMC and DTHMC are partially offset by the increased number of numerical integration steps required to accurately simulate GTHMC trajectories, as seen in ESS per gradients. The minimum ESS came from the mean estimator along the first coordinate for all the simulations, except for DTHMC with $\gamma = 1$ and $T = 25$; in general the directions orthogonal to the tempered one are explored less efficiently by DTHMC as the parameter $\gamma$ and the temperature $T$ increases.
	
\begin{table}
\caption{Comparison of minimum ESS at different temperatures for the 2-d bimodal target. ESS per 100 MCMC samples or per 6656 gradients evaluations are shown.}
\label{table:ess_for_2d_bimodal}
\centering
\vspace{5pt}
\scalebox{0.9}{
\begin{tabular}{rccccc}
\hline                                                                                                                          
Temperature & 5 & 10 & 15 & 20 & 25 \\                                                                                          
\hline                                                                                                                          
ITHMC ESS per samples & 0.279 & 0.421 & 0.445 & 0.469 & 0.510 \\                                                                 
\hline                                                                                                                          
DTHMC ($\gamma = .75$) ESS per samples & 1.10 & 2.56 & 3.20 & 3.67 & 3.63 \\                                                      
\hline                                                                                                                          
DTHMC ($\gamma = 1$) ESS per samples & 3.91 & 13.0 & 17.9 & 18.2 & 16.4 \\  
\hline
NUTS ($\delta = .7$) ESS per samples &  \multicolumn{5}{c}{0.0342} \\                                                           
\hline                                                                                                                          
ITHMC ESS per gradients & 3.37 & 4.90 & 5.11 & 5.27 & 5.80 \\                                                                     
\hline                                                                                                                          
DTHMC ($\gamma = .75$) ESS per gradients & 8.60 & 17.6 & 21.3 & 21.4 & 22.1 \\                                                   
\hline                                                                                                                          
DTHMC ($\gamma = 1$) ESS per gradients & 23.0 & 49.8 & 59.4 & 65.3 & 52.2 \\       
\hline
NUTS ($\delta = .7$) ESS per gradients &  \multicolumn{5}{c}{1} \\     
\hline                                                                                                                          
\end{tabular}
}
\end{table}  

\begin{figure}
	\centering
	\begin{minipage}[b]{0.55\textwidth}
	\includegraphics[width=\textwidth]{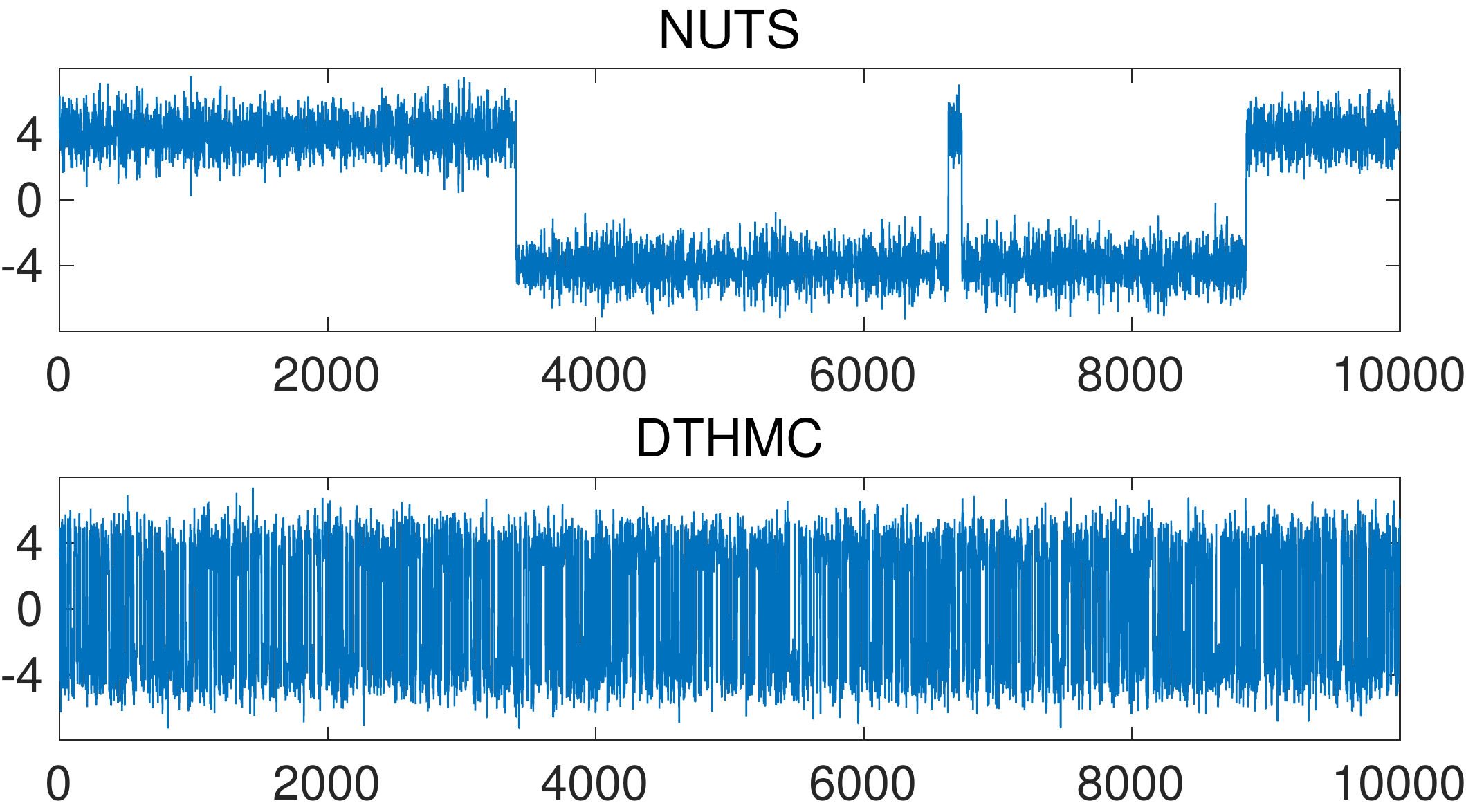} 
	\caption{Traceplot of the first coordinate from $10^4$ samples generated by NUTS ($\delta = 0.7$) and DTHMC ($T = 20, \gamma = 1$).}
	\label{fig:bimodal_traceplot}
	\end{minipage}
	~
	\begin{minipage}[b]{0.38\textwidth}
\centering
\includegraphics[width=\textwidth]{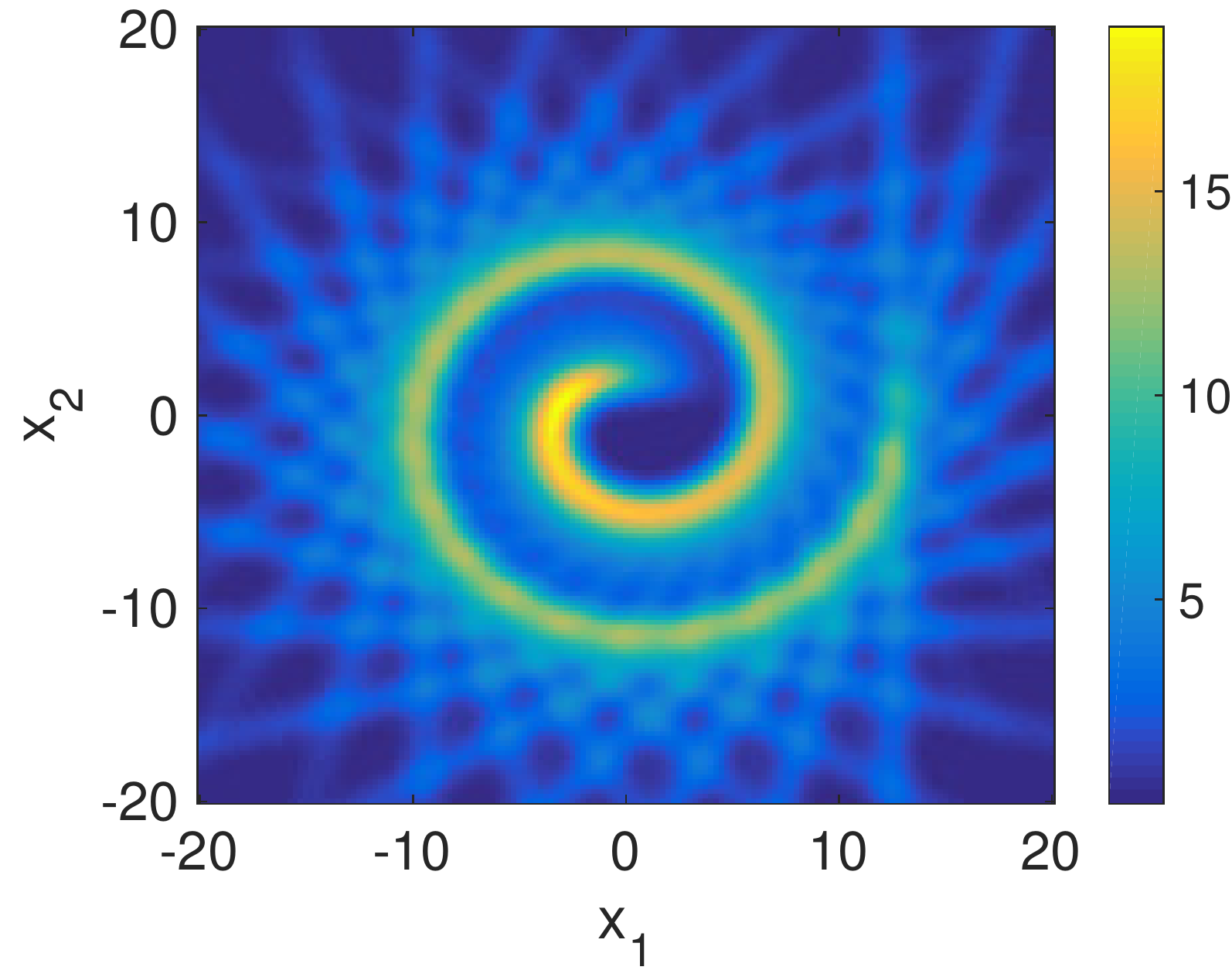} 
\label{fig:rough_swiss_roll}
\caption{Plot of unnormalized swiss roll target distribution.}
	\end{minipage}
	\end{figure}
	
\subsection{Example: Swiss roll distribution}
	For a ``swiss roll'' target as shown in Figure~\ref{fig:rough_swiss_roll}, defined as a Gaussian mixture, we ran NUTS, ITHMC, and DTHMC with $\gamma = .75$. The tempering direction for DTHMC was generated uniformly from a space of unit vectors and independently at each iteration. The performance of each algorithm is summarized in Table~\ref{table:ess_for_swiss_roll}. The potential energy barrier between the ``inner'' and ``outer'' roll is not large, so ITHMC can easily move between them even at $T = 5$. It appears that increasing temperature beyond this point is wasteful in terms of the number of gradient evaluations as the trajectories spend more time exploring the low probability region before finally arriving at the high probability region. It is possible, however, the decrease in ESS per gradients is an artifact of our tuning algorithm. The efficiency of DTHMC here is limited by the lack of preferred direction in the target distribution. 


\begin{table}
\caption{Comparison of ESS across different temperatures for the swiss roll target. ESS per 100 MCMC samples or per 214 gradients evaluations are shown.}
\label{table:ess_for_swiss_roll}                                                                                                                                                                                                                   
\centering                                                                                                                     
\vspace{5pt}
\scalebox{0.9}{
\begin{tabular}{rccccc}                                                                                                 
\hline                                                                                                                         
Temperature & 5 & 10 & 15 & 20 & 25 \\                                                                                         
\hline                                                                                                                         
ITHMC ESS per samples & 50.4 & 42.1 & 42.4 & 46.8 & 42.6 \\   
\hline                                                                                                                         
DTHMC ($\gamma = .75$) ESS per samples & 10.5 & 10.6 & 10.3 & 10.6 & 11.1 \\                                          
\hline
NUTS ($\delta = .8$) ESS per samples &  \multicolumn{5}{c}{6.48} \\
\hline                                                                                                                         
ITHMC ESS per gradients & 1.80 & 1.73 & 1.68 & 1.59 & 1.54 \\
\hline                                                                                                                         
DTHMC ($\gamma = .75$) ESS per gradients & 0.637 & 0.581 & 0.598 & 0.528 & 0.495 \\ 
\hline
NUTS ($\delta = .8$) ESS per gradients &  \multicolumn{5}{c}{1} \\                                 
\hline                                                                                                                         
\end{tabular}                                                                                                                  
}
\end{table}

\subsection{Example: spherically symmetric ``donut'' distribution}
\label{sec:spherical_example}
	To see how GTHMC performs in higher dimensions, we ran NUTS and ITHMC on a 25-dimensional spherically symmetric distribution defined as follows:
	\begin{equation*}
	\pi(\btheta) = \sum_{i=1}^3 \frac{1}{\sigma} \exp \left( - \frac{(\Norm{\btheta} - \mu_i)^2}{2 \sigma^2} \right)
	\ \text{ where } \mu_i = i/2, \ \sigma = 0.1
	\end{equation*}
	The probabilities are therefore concentrated at the spherical shells of radius $\mu_i$'s. One may wonder if the bottleneck in this example is multi-modality or other geometric features, so we additionally report the ESS of a statistic $\Norm{\btheta}$ as a measure of efficiency in exploring the radial direction. The results are summarized in Table~\ref{table:ess_for_spherical}. The ESSs along the radial direction are much smaller, clearly indicating the multimodality to be the bottleneck. Also clear is ITHMC's ability to better deal with the multimodality. In addition, the higher coordinate-wise ESS shows that ITHMC inherits the ability of HMC to explore a complex distribution relatively efficiently. 
	
	The temperature of ITHMC was fixed at $T = 5$ since, as in the swiss roll example, the performance did not change significantly at higher temperature. DTHMC was not tried on this example since DTHMC does not scale well to higher dimensions without localizing the Riemannian metric, which is beyond the scope of this paper.
	
\begin{table}
\caption{Comparison of ESS along the coordinates and along the radial direction. ESS per 100 MCMC samples or per 831 gradient evaluations are shown.}
\label{table:ess_for_spherical}                                                                                                     
\centering                                                                                                                     
\vspace{5pt}
\scalebox{0.9}{
\begin{tabular}{rcc}                                                                                                     
\hline                                                                                                                         
 & Coordinate-wise & Radial \\                                                                                                  
\hline                                                                                                                         
 ITHMC ($T = 5$) ESS per samples & 12.7 & 3.28 \\                                                                                    
\hline                                                                                                                         
 NUTS ($\delta = 0.8$) ESS per samples & 7.30 & 1.13 \\                                                                                    
\hline                                                                                                                         
 ITHMC ($T = 5$) ESS per gradients & 13.1 & 3.43 \\                                                                                
\hline                                                                                                                         
 NUTS ($\delta = 0.8$) ESS per gradients & 6.43 & 1 \\                                                                                
\hline                                                                                                                         
\end{tabular}                                                                                                                  
}
\end{table} 

\section{Discussion}
\label{sec:discussion}
	This paper presented a theoretical and practical framework for alleviating the tendency of HMC to get stuck at local modes. HMC is a general and powerful sampling algorithm widely used in practice, hence addressing its main weakness is of considerable interest. We established the necessary condition on a Riemannian metric and studied the properties of the corresponding Hamiltonian dynamics. In addition, we developed a novel adaptive reversible integrator as well as improved adaptive-rejection mechanism to address the shortcomings of the standard St\"{o}rmer-Verlet.
	
	GTHMC clearly has room for further improvement in two aspects. First, more research effort is needed to develop better numerical integrators for RMHMC and GTHMC applications. Numerical integrators traditionally have been developed to achieve highly accurate trajectories for a long integration time, while in an RMHMC application a required integration time is usually shorter and accuracy is not so important as overall computational efficiency. \cite{blanes14} is one of the first attempts to develop an integrator tailor-made for HMC beyond the standard St\"{o}rmer-Verlet. To our knowledge, the explicit adaptive reversible integrator for non-separable Hamiltonians presented in Section~\ref{sec:explicit_adaptive_integrator} is the first of its kind, and a better numerical integrator can likely be developed with increased research effort in this area.
	
	Second, GTHMC can benefit from a metric more specifically chosen for each multimodal target distribution rather than the generic ones considered in this paper.  ITHMC is a clear improvement over HMC, but is still not efficient in the absolute sense. In fact, it was observed that ITHMC barely performs better than HMC in higher dimensions when modes are isolated (not reported in the paper). This is because a randomly generated trajectory is unlikely to travel in the right direction in a high dimension without encoding more information in the metric. On the other hand, the bi-modal example in Section~\ref{sec:bimodal_example} demonstrates that GTHMC has potential to sample efficiently even from a target distribution with substantial multi-modality.
	
	It is also worth noting that GTHMC can be combined with other (non-geometric) tempering approaches to further promote transitions among the modes in the presence of severe multi-modality. These tempering methods are meta-algorithms and in practice require an additional specification of a transition kernel to sample from tempered distributions $\propto \pi(\btheta)^{1/T_i}$ where the sequence of temperatures $1 = T_1 < T_2 < \ldots < T_M$ must also be specified by a user \citep{earl05, geyer95, marinari92}. 
	The largest temperature $T_M$ must be large enough that the transition kernel can easily induce transitions from one mode to another. Increasing $T_M$ however comes at the cost of increasing the computational time in relating the tempered distribution $\propto \pi(\btheta)^{1/T_M}$ back to the original distribution. For this reason, even within the tempering algorithms it is desirable to use a transition kernel less prone to be stuck at local modes so that the temperatures do not need to be unnecessarily large. GTHMC can provide such a transition kernel, inheriting otherwise desirable characteristics of HMC.

\section{Appendix: Proof of Theorem~\ref{thm:Newton_Hamilton_bijection}}
\label{sec:RMHMC_newton_proof}

For the purpose of the proof, we consider the Jacobians $\Diff\g_{\btheta}$ and $\Diff(\g^{-1})_{\tilde{\btheta}}$ as bijective maps between $\Real[d]$ and $T_{\tilde{\btheta}} M$ rather than non-square matrices, so that the inverse $(\Diff\g_{\btheta})^{-1} = \Diff(\g^{-1})_{\tilde{\btheta}}$ makes sense. One may think of these Jacobians as a square matrix with respect to some basis for $T_{\tilde{\btheta}}M$. One can easily verify that the calculations in the proof are independent of choice of basis. Additionally, for notational convenience we suppress the superscript $M$ from the gradient $\nabla^M$ for a function defined on a manifold $M$. 
	\begin{proof}
	By direct computation, we will prove the equivalence between the differential equations for $(\btheta(t), \p(t))$ and Hamilton's equations with the Hamiltonian $H(\btheta,\p) = -\log \pi(\btheta) + \frac12 \log \Abs{\G(\btheta)} + \frac12 \p^T \G(\btheta)^{-1} \thinspace \p$. Recalling the relations $\btheta = \g^{-1}(\tilbtheta)$, $\Deriv[\tilde{\btheta}]{t} = \tilp$, and $\p = \Diff\g^T_{\tilbtheta} \, \tilp$, we find
	\begin{align*}
	\Deriv[\btheta]{t} 
		= \Diff(\g^{-1})_{\tilde{\btheta}} \Deriv[\tilde{\btheta}]{t} 
		= (\Diff\g_{\btheta})^{-1} \tilde{\p} 
		= (\Diff\g_{\btheta})^{-1} (\Diff\g_{\btheta})^{-T} \p 
		= \G(\btheta)^{-1} \p 
		= \nabla_{\p} H(\btheta, \p)
	\end{align*} 
	The computation for $\Deriv[\p]{t}$ is a bit more involved. First note that
	\begin{align}
	\Deriv[p_i]{t}
		= \Deriv[]{t} \IP{\PDeriv[\g]{\theta_i} (\btheta) }{\tilp} 
		= \IP{ \left(\Diff\PDeriv[\g]{\theta_i}\right)_{\btheta} \Deriv[\btheta]{t} }{\tilp} + \IP{\PDeriv[\g]{\theta_i}}{\nabla_{\tilbtheta} \log \tilde{\pi} } \label{eq:mom_unsimplified}
	\end{align} 
	The first term in the last equation will simplify as follows:
	\begin{align*} 
	\IP{ \left(\Diff\PDeriv[\g]{\theta_i}\right)_{\btheta} \Deriv[\btheta]{t} }{\tilp}
		&= \IP{ \left( \PDeriv[]{\theta_i} \Diff\g_{\btheta} \right) \G(\btheta)^{-1} \p}{\Diff\g_{\btheta} \G(\btheta)^{-1} \p} \\
		&= \frac12 \p^T \G(\btheta)^{-1} \PDeriv[]{\theta_i} \left( \Diff\g_{\btheta}^T \Diff\g_{\btheta}\right) \G(\btheta)^{-1} \p \\
		&= \frac12 \p^T \G(\btheta)^{-1} \PDeriv[\G(\btheta)]{\theta_i} \G(\btheta)^{-1} \p \yesnumber \label{eq:mom_unsimplified_first_term}
	\end{align*}
	We can simplify the second term in \eqref{eq:mom_unsimplified} using Lemma~\ref{lem:gradient_under_transformation} follows:
	\begin{align*}
	\IP{\PDeriv[\g]{\theta_i}}{\nabla_{\tilbtheta} \log \tilde{\pi} }
		&= \IP{(\Diff\g_{\btheta})^{-1} \PDeriv[\g]{\theta_i}}{ \nabla_{\btheta} \log \pi (\btheta) - \frac12 \nabla_{\btheta} \log \Abs{\G(\btheta)}} \\
		&= \PDeriv[]{\theta_i} \log \pi (\btheta) 
			- \frac12 \PDeriv[]{\theta_i} \log \Abs{\G(\btheta)} \yesnumber \label{eq:mom_unsimplified_second_term}
	\end{align*} 
	From \eqref{eq:mom_unsimplified}, \eqref{eq:mom_unsimplified_first_term}, and \eqref{eq:mom_unsimplified_second_term}, we conclude that
	\begin{equation*} \label{eq:mom_deriv}
	\Deriv[p_i]{t}
		= \PDeriv[]{\theta_i} \log \pi (\btheta) 
			- \frac12 \PDeriv[]{\theta_i} \log \Abs{\G(\btheta)}
			+ \frac12 \p^T \G(\btheta)^{-1} \PDeriv[\G(\btheta)]{\theta_i} \G(\btheta)^{-1} \p 
		= - \PDeriv[]{\theta_i} H(\btheta, \p) \qedhere
	\end{equation*}
	\end{proof}
	
	\begin{Lemma}
	\label{lem:gradient_under_transformation}
	If $\pi(\btheta)$ is a pdf on $\Real[d]$ and $\tilde{\pi}(\tilbtheta)$ is a pdf on a manifold $M$ induced by the bijective map $\g:\Real[d] \to M$, then 
	\[ \nabla_{\tilbtheta} \log \tilde{\pi}(\tilbtheta)
		= (\Diff\g_{\btheta})^{-T} \left( \nabla_{\btheta} \log \pi (\btheta) - \frac12 \nabla_{\btheta} \log \left| \Diff \g_\btheta^T \Diff \g_\btheta \right| \right) \]
	\end{Lemma}
	\begin{proof}
	By the change of variable formula, we have
	\begin{equation}
	\log \tilde{\pi}\big( \tilde{\btheta} \big) = 
		- \frac{1}{2} \log \left| \Diff \g_\btheta^T \Diff \g_\btheta \right| + \log \pi(\btheta)
	\end{equation}
	Now we only need to observe that the following equality holds for any scalar-valued function $f(\btheta)$ on $\Real[d]$:
	\begin{equation*}
	\nabla_{\tilbtheta} f \circ \g^{-1} (\tilbtheta)
		= \Diff(\g^{-1})_{\tilbtheta}^T \, \nabla_{\btheta} f(\btheta)
		= (\Diff\g_{\btheta})^{-T}  \nabla_{\btheta} f(\btheta) \qedhere
	\end{equation*}
	\end{proof}

\spacingset{1.45} 
\bibliographystyle{apalike}
\bibliography{GTHMC}

\newpage
\allowdisplaybreaks
\appendix
 \begin{center}
	{\LARGE\bf Supplementary Materials to ``Geometrically Tempered Hamiltonian Monte Carlo''}
\end{center}
\medskip

\section{Relevant geometric notions}

\subsection{Gradient on a manifold}
Consider a function $\tilU(\tilbtheta)$ defined on a $d$-dimensional manifold $M \subset \Real[\tilde{d}]$ and let $T_{\tilbtheta} M \subset \Real[\tilde{d}]$ denote the tangent space of $M$ at $\tilbtheta$. The gradient $\nabla^M \tilU(\tilbtheta)$ can be defined as a unique vector in $T_{\tilbtheta} M$ such that
\begin{equation} \label{eq:def_manif_gradient}
\IP{\nabla^M \tilU(\tilbtheta)}{\mathbf{c}'(0)} 
= \left. \Deriv{t} \thinspace \tilU(\mathbf{c}(t)) \right|_{t=0}
\end{equation}
for all differentiable curves $\mathbf{c}(t)$ on $M$ with $\mathbf{c}(0) = \tilbtheta$. Notice that, under the constraint $\Norm{\mathbf{c}'(0)} = 1$, the left hand side in \eqref{eq:def_manif_gradient} is maximized when $\mathbf{c}'(0)$ is parallel to $\nabla^M \tilU(\tilbtheta)$, agreeing with our intuition of the gradient as the direction of the greatest increase in $\tilU(\tilbtheta)$. 

\subsection{Probability density function on a parametrized manifold}
Due to the difference in the integration theory over a Euclidean space and a manifold, a pdf on a manifold is defined slightly differently from those on a Euclidean space. Here we describe one way to define a pdf on a parametrized manifold through a generalized change of variable formula.

Suppose a random variable $\btheta \in \Real[d]$ has a pdf $\pi(\btheta)$. Given a parametrization (i.e.\ differentiable bijection) $\g$ of a manifold $M$, a random variable $\tilbtheta = \g(\btheta) \in M$ has a pdf
\begin{equation} \label{eq:pdf_on_parametrized_manifold}
\tilde{\pi}(\tilbtheta) = \Abs{\Diff\g_{\btheta}^T \Diff\g_{\btheta}}^{-1/2} \pi (\btheta) 
\quad \text{ where } \btheta = \g^{-1}(\tilbtheta)
\end{equation}
If $\g$ were a bijection between Euclidean spaces and $\Diff\g_{\btheta}$ were a square matrix, then the above formula reduces to the standard change of variable formula, where $\Abs{\Diff\g_{\btheta}}$ is the change of volume factor. More generally, it can be shown that $\Abs{\Diff\g_{\btheta}^T \Diff\g_{\btheta}}^{1/2}$ represents the volume of a $d$-dimensional parallelepiped
\[ P = \left\{ 
\sum_{i=1}^d c_i \PDeriv[\g]{\theta_i}(\btheta) \in T_{\btheta}M \thinspace : \thinspace 
\sum_i c_i \leq 1, \thinspace c_i \geq 0
\right\} \]

\subsection{Bijective map from a dynamics on a manifold to Euclidean space}
Given a parametrization $\g:\Omega \subset \Real[d] \to M$ of a manifold $M \subset \Real[\tilde{d}]$, the $d$ by $\tilde{d}$ matrix  $\Diff \g_{g^{-1}(\tilbtheta)}^T$ is a bijection from the tangent space $T_{\tilbtheta}M \subset \Real[\tilde{d}]$ to $\Real[d]$. This is due to the following elementary fact from linear algebra: given a full rank $\tilde{d} \times d$ matrix $\mathbf{A}$, its transpose $\mathbf{A}^T$ is a bijection from range($\mathbf{A}$) to $\Real[d]$. It then follows that the product map $\g^{-1} \times \Diff \g^T$ defined as
\begin{equation} \label{eq:bijection_bet_RMHMC_HMC}
\g^{-1} \times \Diff \g^T(\tilbtheta,\tilp) = (\g^{-1}(\tilbtheta), \Diff \g_{g^{-1}(\tilbtheta)}^T \tilp)
\end{equation}
is a bijection from a collection of tangent space $\cup_{\tilbtheta \in M} T_{\tilbtheta}M$ to $\Omega \times \Real[d]$. (The collection $TM = \cup_{\tilbtheta \in M} T_{\tilbtheta}M$ is also known as a \textit{tangent bundle}.) Therefore the product map bijectively relates a dynamics on a manifold $M$ to one on a Euclidean space.

\section{Geometric theory of manifold Langevin algorithm}
\label{app:geometric_theory_for_mmala}
Riemann manifold Metropolis adjusted Langevin algorithm (MMALA) is the Langevin dynamics analogue of RMHMC and described by \cite{girolami11} as a potentially useful alternative to RMHMC.	Given a metric $\G(\btheta)$, MMALA generates a proposal by approximating the following SDE for $\btheta = (\theta_1, \ldots, \theta_d)$:
\begin{equation}
\label{eq:reparam_langevin}
\begin{aligned}
\diff \theta_i
&= \frac{1}{2} 
\left\{ \G^{-1}(\btheta) \left( \nabla_{\btheta} \log \pi (\btheta) - \frac12 \nabla_{\btheta} \log \Abs{\G(\btheta)} \right) 
\right\}_i \diff t \\
&\hspace{3em} + \big\{ \G^{-1/2}(\btheta) \diff {\mathbf{B}}_t \big\}_i 
+ \frac12 |\G(\btheta)|^{-1/2} \sum_{j=1}^d \PDeriv{\theta_j} \left\{ |\G(\btheta)|^{1/2} \left(\G^{-1}(\btheta) \right)_{ij}  \right\} \diff t 
\end{aligned}
\end{equation}
where $\mathbf{B}(t)$ is a Brownian motion.	Note that the above equation differs from the one originally presented in \cite{girolami11} which contains a transcription error \citep{xifara14}.

Theorem~\ref{thm:MMALA_as_reparam_MALA} below is a Langevin dynamics analogue of Theorem~\ref{thm:RMHMC_as_reparam_HMC}, establishing a geometric connection between the standard Langevin dynamics \eqref{eq:langevin} and the SDE \eqref{eq:reparam_langevin}. Due to the stochastic nature of Langevin dyanmics, defining it on a manifold through the language of extrinsic geometry turns out to be far more challenging than doing the same for Hamiltonian dynamics \citep{rogers00}. For simplicity, therefore, Theorem~\ref{thm:MMALA_as_reparam_MALA} invokes a stronger assumption than Theorem~\ref{thm:RMHMC_as_reparam_HMC} and assumes that the reparametrization $\g$ is a map between subsets of $\Real[d]$.

\begin{Theorem}[Manifold Langevin as reparametrization]
\label{thm:MMALA_as_reparam_MALA}
Given a pdf $\pi(\btheta)$ on $\Real[d]$, let $\tilde{\pi}$ denote the pdf on a domain $U \subset \Real[d]$ induced by a smooth bijection $\g:\Real[d] \to U$. For the initial condition $\btheta_0 \in \Real[d]$ and $\tilbtheta_0 = g(\btheta_0)$, let $\tilbtheta(t)$ denote a weak solution of the SDE
\begin{equation}
\label{eq:langevin}
\diff \tilbtheta = 
\frac{1}{2} \nabla_{\tilbtheta} \log \tilde{\pi}(\tilbtheta) \diff t
+ \diff \tilde{\mathbf{B}}(t)
\end{equation}
where $\tilde{\mathbf{B}}(t)$ is a Brownian motion. Then the stochastic process $\btheta(t) = \g^{-1}\big(\tilbtheta(t) \big)$ is a weak solution of the SDE \eqref{eq:reparam_langevin} with $\G(\btheta) = \Diff \g^T_{\btheta} \Diff \g_{\btheta}$.
\end{Theorem} 

\begin{proof}
Let $\tilbtheta(t)$ be a solution of the SDE \eqref{eq:langevin}. By Ito's lemma, the stochastic process $\btheta(t) = \g^{-1}\big(\tilbtheta(t) \big)$ solves the following SDE in a weak sense:
\begin{equation}	
\label{eq:SDE_after_Ito_lemma}
\diff \btheta(t)
= \frac{1}{2} \Diff \g^{-1}_{\tilbtheta} \nabla_{\tilbtheta} \log \tilde{\pi} (\tilbtheta) + \Diff \g^{-1}_{\tilbtheta} \diff \mathbf{B}(t) + \frac{1}{2} \Delta_{\tilbtheta} \, \g^{-1}(\tilbtheta) \diff t
\end{equation}
where $\mathbf{B}(t)$ is a Brownian motion and $\Delta_{\tilbtheta} = \sum_i \partial^2 / \partial {\tilde{\theta}_i}^2$ is the Laplacian. Since $\G^{-1} (\btheta) = (\Diff \g^T_{\btheta} \Diff \g_{\btheta})^{-1} = \Diff \g^{-1}_{\tilbtheta} \Diff \g^{-T}_{\tilbtheta}$, we have
\begin{equation}
\label{eq:brownian_motion_equiv}
\G^{-1/2}(\btheta(t)) \left(\mathbf{B}(t+\epsilon) - \mathbf{B}(t) \right) \overset{d}{=} \Diff \g^{-1}( \tilbtheta(t)) \left(\mathbf{B}(t+\epsilon) - \mathbf{B}(t) \right)
\end{equation}
and the term $\Diff \g^{-1}_{\tilbtheta} \diff \mathbf{B}(t)$ in \eqref{eq:SDE_after_Ito_lemma} can equivalently be written as $\G^{-1/2} \big( \btheta \big) \diff \mathbf{B}(t)$. Also rewriting the term $\nabla_{\tilbtheta} \log \tilde{\pi} (\tilbtheta)$ using Lemma~\ref{lem:gradient_under_transformation}, the SDE \eqref{eq:SDE_after_Ito_lemma} can expressed as
\begin{equation}
\label{eq:SDE_rewritten}
\diff \btheta
= \frac{1}{2} \G^{-1}(\btheta) \left( \nabla_{\btheta} \log \pi (\btheta) - \frac12 \nabla_{\btheta} \log \Abs{\G(\btheta)} \right)
+ \G^{-1/2} \big( \btheta \big) \diff \mathbf{B}(t) + \frac{1}{2} \Delta_{\tilbtheta} \, \g^{-1}(\tilbtheta) \diff t
\end{equation}
To express the term $\Delta_{\tilbtheta} \, \g^{-1}(\tilbtheta)$ in terms of $\btheta$, note that 
\[ \nabla_{\tilbtheta} \, \{ \g^{-1}(\tilbtheta) \}_i
= \left( \mathbf{e}_i^T \, \Diff \g^{-1}_{\tilbtheta} \right)^T
= (\Diff \g_{\btheta})^{-T} \mathbf{e}_i	\]
Substituting this to Lemma~\ref{lem:div_under_change_of_coord}, we conclude that
\begin{equation*}
\nabla_{\tilbtheta} \cdot \nabla_{\tilbtheta} \, \{\g^{-1}(\tilbtheta)\}_i
=  |\G(\btheta)|^{-1/2} \sum_{j=1}^d \PDeriv{\theta_j} \left\{ |\G(\btheta)|^{1/2} \left(\G^{-1}(\btheta) \right)_{ij}  \right\} \diff t \qedhere
\end{equation*}
\end{proof}

\begin{Lemma}
\label{lem:div_under_change_of_coord}
If $\g:\btheta \to \tilbtheta$ is a smooth bijection between subsets of $\Real[d]$ and $\bm{v}(\tilbtheta)$ is a vector-valued function, then
\begin{equation}
\nabla_{\tilbtheta} \cdot \bm{v}(\tilbtheta)
= |\G(\btheta)|^{-1/2} \, \nabla_{\btheta} \cdot \left\{ |\G(\btheta)|^{1/2} (\Diff \g_{\btheta})^{-1} \bm{v}(\btheta) \right\}
\end{equation}
where $\G(\btheta) = \Diff \g^T_{\btheta} \Diff \g_{\btheta}$, $\bm{v}(\btheta) := \bm{v} \circ \g(\btheta)$, and $\nabla_{\btheta} \, \cdot  = \sum_i \partial / \partial \theta_i$ is the divergence operator.
\end{Lemma}
\begin{proof}
The proof only requires elementary calculus, but the computation is lengthy, involved and hence is omitted here. The details can be found in, for example, Chapter 3 of \cite{leonhardt10}.
\end{proof}

\newcommand{\h}{\mathbf{h}}
\section{Explicit adaptive integrator: further details}
\label{app:explicit_reversible_integrator}
Here we provide further details on the derivation and the properties of the explicit adaptive integrator described in Section~\ref{sec:explicit_adaptive_integrator}. 

\subsection{Derivation of Equation \eqref{eq:vel_diff_eq_noncompact}  }
We first show how one can derive the differential equation \eqref{eq:vel_diff_eq_noncompact} for the parameters $(\btheta,\v)$ from \eqref{eq:time_rescaled_rmhmc}. Similar calculations in the case $\eta(\btheta) \equiv 1$ are carried out in \cite{lan15} and \cite{fang14}. Letting $\h(\btheta,\p) = (\btheta, \eta(\btheta) \G^{-1}(\btheta) \p)$ denote the change of variable from $(\btheta,\p)$ to  $(\btheta,\v)$, we have
\begin{equation} \label{subeq:flow_under_change_of_var}
\Deriv{s} (\btheta(s),\v(s)) 
= \Diff \h {(\btheta(s),\p(s))} \Deriv{s} (\btheta(s),\p(s))
\end{equation}
It is not difficult to show that the Jacobian $\Diff \h$ is given in terms of the variable $(\btheta,\v)$ as:
\begin{equation} \label{subeq:Jacobian_change_of_var}
\Diff \h
= \begin{bmatrix}
\I & \mathbf{0} \\ 
- \eta \G^{-1} \left( \sum_{i} \PDeriv{\theta_i} ( \frac{1}{\eta} \G ) \v \bm{e}_k^T \right) &  \eta \G^{-1}
\end{bmatrix}
\end{equation}	 
By plugging \eqref{subeq:Jacobian_change_of_var} and \eqref{subeq:flow_under_change_of_var} into the differential equation \eqref{eq:time_rescaled_rmhmc} for $(\btheta,\p)$, we obtain
\begin{equation} \label{subeq:lagrangian_dynamics}
\begin{aligned}
\Deriv[\btheta]{s}
&= \v, \quad
\Deriv[\v]{s}
&= - \eta^2 \G^{-1} \nabla_{\btheta} \phi 
- \eta \sum_{i=1}^d v_i \G^{-1} \PDeriv{\theta_i} \Big( \frac{1}{\eta} \G \Big) \v 
+ \frac12 \G^{-1} \v^T (\nabla_{\btheta} \G) \v
\end{aligned}
\end{equation}
With straightforward algebra, the expression for $\diff \v / \diff s$ can be re-written as:
\begin{equation}
\label{subeq:vel_deriv_before_symmetrization}
\begin{aligned}
\Deriv[v_k]{s}
&= -\eta^2 \left[ \G^{-1} \nabla_{\btheta} \phi \right]_k + \v^T \bGamma^k \v \\
& \text{ where } \ \bGamma^k_{ij} = \sum_\ell (\G^{-1})_{k\ell} \left[ \frac12 \PDeriv{\theta_\ell} G_{ij} - \eta \PDeriv{\theta_i} \left( \frac{1}{\eta} G_{\ell j} \right) \right]
\end{aligned}
\end{equation}
Since $\v^T \bGamma^k \v = \v^T (\bGamma^k)^T \v$, we can replace $\bGamma^k$ with its symmetrization $\frac12 (\bGamma^k + (\bGamma^k)^T)$ without changing the equation \eqref{subeq:vel_deriv_before_symmetrization}. So we re-define $\bGamma^k$ to be a matrix such that
\[ \bGamma^k_{ij}
= \sum_\ell (\G^{-1})_{k\ell} \left[ \frac12 \PDeriv{\theta_\ell} G_{ij} 
- \frac{\eta}{2} \PDeriv{\theta_i} \left( \frac{1}{\eta} G_{\ell j} \right)
- \frac{\eta}{2} \PDeriv{\theta_j} \left( \frac{1}{\eta} G_{\ell i} \right) \right] \]
Although the symmetrization of $\bGamma^k$ does not alter the differential equation at all, it will guarantee $(\v^*)^T \bGamma^k \v = \v^T \bGamma^k \v^*$ for all $\v$ and $\v^*$ --- a crucial property in ensuring the reversibility of our explicit adaptive integrator. Finally, if we let $\v^T \bGamma$ denote a matrix whose $k$-th row is given by $\v^T \bGamma^k$, we can express the differential equation \eqref{subeq:lagrangian_dynamics} in the following form, which agrees with \eqref{eq:vel_diff_eq_noncompact}:
\begin{equation} \label{subeq:vel_diff_eq_compact}
\begin{aligned}
\Deriv[\btheta]{s}
&= \v, \quad
\Deriv[\v]{s}
&= -\eta^2 \G^{-1} \nabla \phi + \v^T \bGamma \v
\end{aligned}
\end{equation}

\newcommand{\Fv}[1]{\mathbf{F}_{\v,#1}}
\newcommand{\Ft}[1][\epsilon]{\mathbf{F}_{\btheta,#1}}
\subsection{Reversible explicit discretization}
We now describe how to obtain the explicit reversible integrator \eqref{eq:explicit_integrator_updates} of the differential equation \eqref{subeq:vel_diff_eq_compact}. We also derive the formula for the Jacobian of the integrator, which is needed to calculate the acceptance probability of the variable-length trajectory CHMC algorithm in Section~\ref{app:vlt_chmc}. A reversible explicit update $\v \to \v^*$ is obtained by the following discretization based on a linearly implicit scheme of Kahan \citep{lan15, ss94}:
\begin{align} 	
&\frac{\v^* - \v}{\epsilon}
= - \eta^2 \G^{-1} \nabla_{\btheta} \phi + \v^T \bGamma \v^* \label{eq:vel_discretization} \\
\iff & \v^* = \left( \I - {\epsilon} \v^T \bGamma \right)^{-1} \left( \v - {\epsilon} \eta^2 \G^{-1} \nabla_{\btheta} \phi \right) \label{subeq:vel_update_formula}
\end{align}
Now let $\Fv{\epsilon}$ denote the map $\Fv{\epsilon}(\btheta,\v) = (\btheta, \v^*)$ corresponding to the update equation \eqref{subeq:vel_update_formula}. Note that $\Fv{\epsilon}$ is reversible thanks to the symmetry $\v^T \bGamma \v^* = (\v^*)^T \bGamma \v$. The Jacobian of the map $\v \to \v^*$ is obtained by differentiating Equation~\eqref{eq:vel_discretization} implicitly in $\v$:
\begin{align}
& \frac{\PDeriv[\v^*]{\v} - \I}{\epsilon}
= \v^T \bGamma \PDeriv[\v^*]{\v} + (\v^*)^T \bGamma \nonumber \\
\iff & \PDeriv[\v^*]{\v} = \left( \I - \frac{\epsilon}{2} \v^T \bGamma \right)^{-1} \left( \I + \frac{\epsilon}{2} (\v^*)^T \bGamma \right) \label{eq:v_Jacobian}
\end{align}
A reversible explicit update for $\btheta$ is given by a map $\Ft(\btheta,\v) = (\btheta + \epsilon \v,\v)$, which is obviously reversible and volume preserving. The integrator \eqref{eq:explicit_integrator_updates} is obtained by the composition $\Fv{\epsilon/2} \circ \Ft \circ \Fv{\epsilon/2}$, which is reversible and explicit because both $\Fv{\epsilon/2}$ and $\Ft$ are.

\subsection{Derivation of explicit adaptive integrator for DTHMC}
Here we derive the necessary formulas to carry out an efficient implementation of the integrator \eqref{eq:explicit_integrator_updates} in DTHMC settings. In particular, we show how to simplify the formula of $1 - \frac{\epsilon}{2} \v^T \bGamma$; the rest of the quantities in \eqref{eq:explicit_integrator_updates} are relatively straightforward to compute. To find a formula for the matrix $\bGamma^k$ as defined in \eqref{eq:bGamma_sym}, we start by computing the last two terms of ${\rm d} \v / {\rm d} s$ in \eqref{subeq:lagrangian_dynamics} namely the term $\frac12 \G^{-1} \v^T (\nabla_{\btheta} \G) \v$ and $\eta \sum_{i=1}^d v_i \G^{-1} \PDeriv{\theta_i} \Big( \frac{1}{\eta} \G \Big) \v$. Observe that
\begin{alignat*}{2}
&&\v^T \partial_{\theta_i} \G \thinspace \v
&= \IP{\u}{\v}^2 \partial_{\theta_i} g_\parallel - \IP{\u}{\v}^2 \partial_{\theta_i} g_\perp + \Norm{\v}^2 \partial_{\theta_i} g_\perp \\
\Longrightarrow  \
&&\v^T \nabla \G \thinspace \v
&= \IP{\u}{\v}^2 (\nabla g_\parallel - \nabla g_\perp) + \Norm{\v}^2 \nabla g_\perp \\
\Longrightarrow \
&&\G^{-1} \v^T \nabla \G \thinspace \v
&= \IP{\u}{\v}^2 \left( \left(\frac{1}{g_\parallel}-\frac{1}{g_\perp} \right) \IP{\nabla g_\parallel - \nabla g_\perp}{\u} \u + \frac{1}{g_\perp} (\nabla g_\parallel - \nabla g_\perp) \right)  \\
&& 
&\qquad + \Norm{\v}^2 \left( \left(\frac{1}{g_\parallel}-\frac{1}{g_\perp} \right) \IP{\nabla g_\perp}{\u} \u + \frac{1}{g_\perp} \nabla g_\perp \right) \\
\Longrightarrow \
&& \left( \frac12 \G^{-1} \v^T \nabla \G \thinspace \v \right)_k
&= \v^T \left[ \left( \left(\frac{1}{g_\parallel}-\frac{1}{g_\perp} \right) \IP{\nabla g_\parallel - \nabla g_\perp}{\u} u_k + \frac{1}{g_\perp} (\partial_{\theta_k} g_\parallel - \partial_{\theta_k} g_\perp) \right) \frac12 \u \u^T \right] \v \\
&&
& \qquad + \v^T \left[ \left( \left(\frac{1}{g_\parallel}-\frac{1}{g_\perp} \right) \IP{\nabla g_\perp}{\u} u_k + \frac{1}{g_\perp} \partial_{\theta_k} g_\perp \right) \frac12 \I \right] \v
\end{alignat*}
For the other term, we have
\begin{alignat*}{2}
&& \PDeriv{\theta_i} \left( \frac{1}{\eta} \G \right)
&= \left( \PDeriv{\theta_i}\frac{1}{\eta} \right) \G + \frac{1}{\eta} \left( \left(\partial_{\theta_i} g_\parallel - \partial_{\theta_i} g_\perp \right) \u \u^T + \partial_{\theta_i} g_\perp \, \I \right) \\
\Longrightarrow \
&& \eta \PDeriv{\theta_i} \left( \frac{1}{\eta} \G \right) \v 
&= \eta \left( \PDeriv{\theta_i} \frac{1}{\eta} \right) \G \v
+ \IP{\u}{\v} \left(\partial_{\theta_i} g_\parallel - \partial_{\theta_i} g_\perp \right) \u 
+ \partial_{\theta_i} g_\perp \, \v  \\
\Longrightarrow \
&& \eta \, \G^{-1} \PDeriv{\theta_i} \left( \frac{1}{\eta} \G \right) \v
&= \eta \left( \PDeriv{\theta_i} \frac{1}{\eta} \right) \v
+ \frac{1}{g_\parallel} \IP{\u}{\v} \left(\partial_{\theta_i} g_\parallel - \partial_{\theta_i} g_\perp \right) \u \\
&&
&\qquad \quad + \left( \frac{1}{g_\parallel} - \frac{1}{g_\perp} \right) \partial_{\theta_i} g_\perp \IP{\u}{\v} \u + \frac{1}{g_\perp} \partial_{\theta_i} g_\perp \, \v \\
&&
&= \eta \left( \PDeriv{\theta_i} \frac{1}{\eta} \right) \v
+ \frac{1}{g_\parallel} \IP{\u}{\v} \left(\partial_{\theta_i} g_\parallel \right) \u \\
&&
& \qquad \quad
- \frac{1}{g_\perp} \partial_{\theta_i} g_\perp \IP{\u}{\v} \u + \frac{1}{g_\perp} \partial_{\theta_i} g_\perp \, \v \\
\Longrightarrow \
&& \eta \sum_{i} v_i \G^{-1} \PDeriv{\theta_i} \left( \frac{1}{\eta} \G \right) \v &= \\ && & \hspace{-7em}
\eta \IP{\v}{\nabla \frac{1}{\eta}} \v 
+ \IP{\u}{\v} \IP{\v}{\nabla (\log g_\parallel - \log g_\perp)} \u + \IP{\v}{\nabla \log g_\perp} \, \v \\
\Longrightarrow \
&& \left( -\eta \sum_{i} v_i \G^{-1} \PDeriv{\theta_i} \left( \frac{1}{\eta} \G \right) \v \right)_k &=
\\* && & \hspace{-7em}
\v^T \left[
- \eta \nabla \frac{1}{\eta} \cdot \e_k^T 
- \u \cdot \nabla^T (\log g_\parallel - \log g_\perp) \, u_k 
- \nabla \log g_\perp \cdot \e_k^T \right] \v
\end{alignat*}
Since $\eta = \sqrt{g_\parallel}$, we have $\eta \nabla 1/\eta = - \frac12 \nabla \log g_\parallel$ and therefore
\begin{align*}
\left( -\eta \sum_{i} v_i \G^{-1} \PDeriv{\theta_i} \left( \frac{1}{\eta} \G \right) \v \right)_k & \\
& \hspace{-7em}
= \v^T \left[
\frac12 (\nabla \log g_\parallel - 2 \nabla \log g_\perp) \cdot \e_k^T 
- \u \cdot \nabla^T (\log g_\parallel - \log g_\perp) \, u_k 
\right] \v
\end{align*}
Thus the (symmetrized) matrix $\bGamma^k$ must be given by
\begin{align*}
\bGamma^k
&= \frac12 \left( \left(\frac{1}{g_\parallel}-\frac{1}{g_\perp} \right) \IP{\nabla g_\parallel - \nabla g_\perp}{\u} u_k + \frac{1}{g_\perp} (\partial_{\theta_k} g_\parallel - \partial_{\theta_k} g_\perp) \right)  \u \u^T  \\
&\qquad + \frac12 \left( \left(\frac{1}{g_\parallel}-\frac{1}{g_\perp} \right) \IP{\nabla g_\perp}{\u} u_k + \frac{1}{g_\perp} \partial_{\theta_k} g_\perp \right) \I \\
&\qquad + \frac14 (\nabla \log g_\parallel - 2 \nabla \log g_\perp) \cdot \e_k^T 
+ \frac14 \e_k \cdot (\nabla \log g_\parallel - 2 \nabla \log g_\perp)^T \\
&\qquad - \frac12 \u \cdot \nabla^T (\log g_\parallel - \log g_\perp) \, u_k
- \frac12 \nabla (\log g_\parallel - \log g_\perp) \cdot \u^T \, u_k 
\end{align*}
From this formula it easily follows that
\begin{align*}
\v^T \bGamma
&= \frac12 \left(\frac{1}{g_\parallel}-\frac{1}{g_\perp} \right) \IP{\nabla g_\parallel - \nabla g_\perp}{\u} \IP{\v}{\u} \u \u^T + \frac{1}{2h} \IP{\v}{\u} (\nabla g_\parallel - \nabla g_\perp) \cdot \u^T  \\
&\qquad + \frac12 \left(\frac{1}{g_\parallel}-\frac{1}{g_\perp} \right) \IP{\nabla g_\perp}{\u} \u \v^T + \frac{1}{2} \nabla \log g_\perp \cdot \v^T \\
&\qquad + \frac14 \IP{\v}{\nabla \log g_\parallel - 2 \nabla \log g_\perp} \I 
+ \frac14 \v \cdot (\nabla \log g_\parallel - 2 \nabla \log g_\perp)^T \\
&\qquad - \frac12 \IP{\v}{\u} \u \cdot \nabla^T (\log g_\parallel - \log g_\perp)
- \frac12 \IP{\v}{\nabla (\log g_\parallel - \log g_\perp)}  \u \u^T  \\
&= \frac12 \left(1-\frac{g_\parallel}{g_\perp} \right) \IP{\nabla \log g_\parallel}{\u} \IP{\v}{\u} \u \u^T + \frac{1}{2} \IP{\v}{\u} \left(\frac{g_\parallel}{g_\perp} \nabla \log g_\parallel - \nabla \log g_\perp \right) \cdot \u^T \\
&\qquad + \frac12 \left(\frac{g_\perp}{g_\parallel}-1 \right) \IP{\nabla \log g_\perp}{\u} \u \cdot (\v^T - \IP{\v}{\u} \u^T) + \frac{1}{2} \nabla \log g_\perp \cdot \v^T \\
&\qquad + \frac14 \IP{\v}{\nabla \log g_\parallel - 2 \nabla \log g_\perp} \I 
+ \frac14 \v \cdot (\nabla \log g_\parallel - 2 \nabla \log g_\perp)^T \\
&\qquad - \frac12 \IP{\v}{\u} \u \cdot \nabla^T (\log g_\parallel - \log g_\perp)
- \frac12 \IP{\v}{\nabla (\log g_\parallel - \log g_\perp)}  \u \u^T
\end{align*}
To express $\v^T \Gamma$ as a low-rank perturbation of identity, we first note that $\log g_\perp / \log g_\parallel = c $ where $c = \frac{1-\gamma}{\gamma (d-1)}$. Using this relation, we have the following three equalities:
\begin{align*}
&\frac14 \v \cdot (\nabla \log g_\parallel - 2 \nabla \log g_\perp)^T - \frac12 \IP{\v}{\u} \u \cdot \nabla^T (\log g_\parallel - \log g_\perp) 
\\* & \hspace{5em} 
= \left(\frac14 (1-2c) \v - \frac12 (1-c) \IP{\v}{\u} \u \right) \cdot \nabla^T \log g_\parallel 
\\
&\frac{1}{2} \IP{\v}{\u} \left(\frac{g_\parallel}{g_\perp} \nabla \log g_\parallel - \nabla \log g_\perp \right) \cdot \u^T
+ \frac12 \nabla \log g_\perp \cdot \v^T \\
&\hspace{5em} = \nabla \log g_\parallel \cdot \left( \frac{1}{2} \IP{\v}{\u} \left(\frac{g_\parallel}{g_\perp}-c\right) \u^T + \frac{c}{2} \v^T \right)
\end{align*}
\begin{align*}
&\frac12 \left(1-\frac{g_\parallel}{g_\perp} \right) \IP{\nabla \log g_\parallel}{\u} \IP{\v}{\u} \u \u^T \\*
&\hspace{8em} + \frac12 \frac{g_\perp}{g_\parallel} \left(1 - \frac{g_\parallel}{g_\perp} \right) \IP{\nabla \log g_\perp}{\u} \u \cdot (\v^T - \IP{\v}{\u} \u^T) \\
&\hspace{8em} - \frac12 \IP{\v}{\nabla (\log g_\parallel - \log g_\perp)}  \u \u^T \\
&\hspace{5em} = \frac12 \u \cdot \bigg[ \left(1-\frac{g_\parallel}{g_\perp} \right) \IP{\nabla \log g_\parallel}{\u} \left(\IP{\v}{\u} \u + \frac{c g_\perp}{g_\parallel} (\v - \IP{\v}{\u} \u) \right) \\
&\hspace{25em} - (1-c) \IP{\v}{\nabla \log g_\parallel} \u \bigg]^T
\end{align*}
So the formula for $\v^T \bGamma$ can be simplified as
\begin{align*}
\v^T \bGamma
&= \frac14 (1-2c) \IP{\v}{\nabla \log g_\parallel} \I \\
&\qquad + \left(\frac14 (1-2c) \v - \frac12 (1-c) \IP{\v}{\u} \u \right) \cdot \nabla^T \log g_\parallel \\
&\qquad + \nabla \log g_\parallel \cdot \left( \frac{1}{2} \IP{\v}{\u} \left(\frac{g_\parallel}{g_\perp}-c\right) \u^T + \frac{c}{2} \v^T \right) \\
&\qquad + \frac12 \u \cdot \bigg[ \left(1-\frac{g_\parallel}{g_\perp} \right) \IP{\nabla \log g_\parallel}{\u} \left(\IP{\v}{\u} \u + \frac{c g_\perp}{g_\parallel} (\v - \IP{\v}{\u} \u) \right) \\
&\hspace{20em} - (1-c) \IP{\v}{\nabla \log g_\parallel} \u \bigg]^T
\end{align*}
And finally we obtain
\begin{align*}
1 - \frac{\epsilon}{2} \v^T \bGamma
&= \left(1 - \frac{\epsilon}{8} (1-2c) \IP{\v}{\nabla \log g_\parallel} \right) \I \\
&\quad - \epsilon \left(\frac18 (1-2c) \v - \frac14 (1-c) \IP{\v}{\u} \u \right) \cdot \nabla^T \log g_\parallel \\
&\quad - \frac{\epsilon}{2} \nabla \log g_\parallel \cdot \left( \frac{1}{2} \IP{\v}{\u} \left(\frac{g_\parallel}{g_\perp}-c\right) \u^T + \frac{c}{2} \v^T \right) \\
&\quad - \frac{\epsilon}{4} \u \cdot \bigg[ \left(1-\frac{g_\parallel}{g_\perp} \right) \IP{\nabla \log g_\parallel}{\u} \left(\IP{\v}{\u} \u + \frac{c g_\perp}{g_\parallel} (\v - \IP{\v}{\u} \u) \right) \\
&\hspace{20em} - (1-c) \IP{\v}{\nabla \log g_\parallel} \u \bigg]^T
\end{align*}

\section{Variable length trajectory compressible HMC}
\label{app:vlt_chmc}
The explicit adaptive integrator of Section~\ref{sec:explicit_adaptive_integrator} is not volume-preserving and hence the standard acceptance-rejection scheme yields an incorrect stationary distribution. Compressible HMC (CHMC) modifies the acceptance probabilities appropriately to preserve the original target distribution. In GTHMC settings, however, CHMC in general suffers from low acceptance probabilities and poor mixing. The issue is that Hamiltonian dynamics no longer preserves the original target distribition after time-rescaling \eqref{eq:time_rescaled_Hamilton} \citep{nishimura16vlt}. Variable length trajectory compressible HMC (VLT-CHMC) constructs a transition kernel that better approximates the original dynamics by allowing individual trajectories to have different path lengths.  We focus on the motivations and main ideas behind the algorithm and highlight its advantage over the standard CHMC under GTHMC settings.  More thorough presentation and technical details are available in \cite{nishimura16vlt}. 

Let $\Fe$ denote the reversible bijective map $(\btheta_0,\v_0) \to (\btheta_1, \v_1)$ as defined in \eqref{eq:explicit_integrator_updates}.  CHMC
would use the map $\R \circ \Fe^n$ to generate a proposal where $\Fe^n = \Fe \circ \ldots \circ \Fe$ for $n \in \mathbb{Z}^+$ and $\R(\btheta,\v) = (\btheta, - \v)$. 
The acceptance rate of a proposal $(\btheta^*,\v^*) = \R \circ \Fe^n(\btheta_0,\v_0)$ tends to be low since the map $\Fe^n$ approximates the solution $(\btheta_0, \v_0) \to {(\btheta(s), \v(s))}$ for $s = n \epsilon$ of the time-rescaled dynamics \eqref{eq:vel_diff_eq_noncompact} and not of the original dynamics. In particular, the acceptance probability converges to $1 \wedge \eta(\btheta(\varsigma)) / \eta(\btheta_0)$ in the limit $\epsilon \to 0$ with $n \epsilon \to \varsigma > 0$ \citep{nishimura16vlt}.

On the other hand, VLT-CHMC constructs a transition kernel that better approximates the dynamics in the original time scale as follows. 
From the relation $\eta(\btheta) \diff s = \diff t$, it follows that solving the time-rescaled dynamics for time $s$ is equivalent to solving the original dynamics for time $t = \int_0^s \eta(\btheta(s')) \diff s'$. Therefore the map $(\btheta_0, \v_0) \to (\btheta, \v)(t = \tau)$ can be approximated by the map $\Fe^N(\btheta, \p) := \Fe^{N(\btheta, \p)}(\btheta, \p)$ where, denoting $(\btheta_i, \v_i) = \Fe^i(\btheta_0, \v_0)$,
\begin{equation}
\label{eq:def_variable_step_num}
N(\btheta_0,\v_0)  = N(\tau, \epsilon, \btheta_0, \v_0) = \min_n \left\{ n : \tau < \sum_{i=1}^n \, \epsilon \, \frac{\eta(\btheta_{i-1}) + \eta(\btheta_i)}{2}  \right\}
\end{equation}
The map $\Fe^N$ in general is not reversible and hence the map $(\btheta, \v) \to (\btheta^*, \v^*) = \R \circ \Fe^N(\btheta,\v)$ does not constitute a valid proposal move under the CHMC framework. However, it is possible to construct collections of states $S$ and $S^*$ containing $(\btheta,\v)$ and $(\btheta^*, \v^*)$ such that
\begin{equation}
\label{eq:reversibility_bet_sets}	
\begin{aligned}
\R \circ \Fe^N(S) \subset S^* 
&\ \text{ and } \ \R \circ \Fe^N(S^*) \subset S \\
\R \circ \Fe^N(S^c) \subset (S^{*})^c 
&\ \text{ and } \
\R \circ \Fe^N \left((S^{*})^c \right) \subset S^c
\end{aligned}
\end{equation}
A reversible Markov chain can then be obtained by proposing the transition from the collection of states $S$ to $S^*$ and vice versa.

Unlike CHMC ones, VLT-CHMC proposals are guaranteed high acceptance probabilities. Also, each iteration of VLT-CHMC requires little additional computation beyond what it takes to approximate a trajectory of the dynamics (in the original time scale). These facts are made precise in the following theorem. The proof and the empirical comparison between CHMC and VLT-CHMC can be found in \cite{nishimura16vlt}.
\begin{Theorem}[VLT-CHMC] \label{thm:VLT_CHMC}
Given a reversible integrator with stepsize $\epsilon$ of a time-rescaled Hamiltonian dynamics of the form~(\ref{eq:time_rescaled_Hamilton}), VLT-CHMC produces a reversible transition kernel with the following properties. In the statements below, a proposal generated from $(\btheta_0, \p_0)$ is considered and $(\btheta(\tau), \p(\tau))$ denotes the exact solution of Hamiltonian dynamics at time $\tau$ in the original time-scale: 
\begin{enumerate}[(a).]
	\item (High acceptance probability) For $\eta_0^* = \eta(\btheta(\tau))$ and $\eta_0 = \eta(\btheta_0)$, the acceptance probability of the transition to $S^*$ as $\epsilon \to 0$ converges to a value bounded below by $\frac{\eta_0^*}{\eta_0} \floor{\frac{\eta_0}{\eta_0^*}}$ if $\eta_0^* < \eta_0$ and by $\frac{\eta_0}{\eta_0^*} \floor{\frac{\eta_0^*}{\eta_0}}$ otherwise.
	\item The number of integration steps required for the proposal generation is given by  $$ N(\tau, \btheta_0, \p_0) +  \max \left\{ \floor{\frac{\eta_0}{\eta_0^*}}, \floor{\frac{\eta_0^*}{\eta_0}} \right\} + 1 \pm O(\epsilon) $$
	where $N$ is the step number function as in \eqref{eq:def_variable_step_num}. 
\end{enumerate}
\end{Theorem}

\end{document}